%% file: main.tex
\newcommand{\frameworkname}{CHAT} 
\newcommand{\starefcm}{$\textit{efc}^*_{M}$}

\newcommand{\M}{$\textit{M}$}
\newcommand{\efc}{$\textit{efc}$}
\newcommand{\efs}{$\textit{efs}$}

\documentclass[sigconf]{acmart}
\usepackage{graphicx} 
\usepackage{hyperref}
\usepackage{cleveref}
\usepackage{enumitem}
\usepackage{bbm}
\usepackage{booktabs}
\usepackage{tabularx}
\usepackage{adjustbox}
\usepackage{caption}
\usepackage{subcaption}
\usepackage{makecell}
\usepackage{algorithm}
\usepackage{algpseudocode}
\usepackage{amsmath}
\usepackage{tikz}
\usepackage{multirow}

\usepackage{parskip}
\newcommand{\blackcircled}[1]{%
  \tikz[baseline=(char.base)]{
    \node[shape=circle, fill=black, inner sep=0.5pt] (char) {\textcolor{white}{#1}};
  }%
}

\algnewcommand\algorithmicinput{\textbf{Input:}}
\algnewcommand\algorithmicoutput{\textbf{Output:}}
\algnewcommand\Input{\item[\algorithmicinput]}
\algnewcommand\Output{\item[\algorithmicoutput]}

\newlength{\subfigHis}
\AtBeginDocument{%
  \providecommand\BibTeX{{%
    Bib\TeX}}}

\setcopyright{acmlicensed}
\copyrightyear{2027}
\acmYear{2027}
\acmDOI{XXXXXXX.XXXXXXX}
\acmConference[SIGMOD '27]{Proceedings of the 2027 International Conference on Management of Data}{June 13--19, 2027}{Huntington Beach, CA, USA}
\acmBooktitle{Proceedings of the 2027 International Conference on Management of Data (SIGMOD '27), June 13--19, 2027, Huntington Beach, CA, USA}
\acmISBN{978-1-4503-XXXX-X/2018/06}

\begin{document}
\title{\jae{Exploiting Structural Properties for Efficient Constraint-Aware HNSW Hyperparameter Tuning}}


\author{Geon Choi}
\orcid{0000-0000-0000}
\affiliation{%
  \institution{Seoul National University}
  \country{South Korea}
  }
\email{cds06126@snu.ac.kr}

\author{Hoeun Lee}
\orcid{0000-0000-0000}
\affiliation{%
  \institution{Seoul National University}
  \country{South Korea}
  }
\email{hoeunlee@snu.ac.kr}

\author{Jaeyoung Do}
\orcid{0000-0000-0000}
\authornote{Corresponding Author.}
\affiliation{%
  \institution{Seoul National University}
  \country{South Korea}
  }
\email{jaeyoung.do@snu.ac.kr}

\renewcommand{\shortauthors}{Choi et al.}


\definecolor{comportableColor}{HTML}{2370CD}
\definecolor{todoRed}{HTML}{F05650}
\definecolor{goodGreen}{HTML}{006E51}
\newcommand{\jae}[1]{\textcolor{black}{#1}}
\newcommand{\todo}[1]{{\color{black}#1}}
\newcommand{\hoeun}[1]{{\color{black}#1}}
\newcommand{\gun}[1]{{\color{black}#1}}
\newcommand{\reviewerA}[1]{{\color{black}#1}}
\newcommand{\reviewerB}[1]{{\color{black}#1}}
\newcommand{\reviewerC}[1]{{\color{black}#1}}
\newcommand{\cgcg}[1]{{\color{black}#1}}
\newcommand{\revisionTodo}[1]{{\color{black}#1}}

\input{sections/abstract}



\received{20 February 2007}
\received[revised]{12 March 2009}
\received[accepted]{5 June 2009}

\maketitle

\input{sections/introduction}

\input{sections/background}
\input{sections/hyperparameters_performance}
\input{sections/hyperparameters_resource}
\input{sections/solution}
\input{sections/experiment}

\input{sections/related_work}
\input{sections/conclusion}

\bibliographystyle{ACM-Reference-Format}
\bibliography{sample-base}
\clearpage
\appendix
\input{appendix}

\end{document}

%% file: sections/abstract.tex
\begin{abstract}
Vector databases (VectorDBs) are a core component of modern retrieval systems, including Retrieval-Augmented Generation (RAG), where efficient Approximate Nearest Neighbor Search (ANNS) is critical. Among ANNS algorithms, Hierarchical Navigable Small World (HNSW) graphs are widely adopted for their strong recall-latency trade-off. However, configuring HNSW remains challenging: its hyperparameters jointly affect search quality, latency, build time, and index size in nonlinear ways, while production deployments impose strict resource and tuning-time constraints.
We study HNSW hyperparameter tuning from a systems perspective and show that its configuration space exhibits strong structural regularities. Specifically, we identify monotonic, dominant unimodal, and separable relationships among search-time and construction-time parameters, which induce feasibility boundaries under performance and resource constraints. Building on this insight, we propose \frameworkname{}, a constraint-aware tuning framework for HNSW. Unlike generic black-box optimizers, \frameworkname{} exploits HNSW-specific structure to perform deterministic, sample-efficient search and prune resource-infeasible configurations before full index construction.
Across multiple datasets and HNSW-based vector search engines, \frameworkname{} identifies configurations that maximize recall or throughput while satisfying constraints on accuracy, latency, build time, index size, and tuning budget. Compared to strong baselines, \frameworkname{} achieves up to 45\% higher throughput or 11\% higher recall, and converges up to $44\times$ faster. These results show that principled, structure-aware tuning enables efficient and robust HNSW deployment beyond generic black-box optimization.

\end{abstract}

%% file: sections/introduction.tex
\section{\reviewerA{Introduction}}

Recent advances in large language models (LLMs) have led to strong performance in text-based retrieval, generation, and summarization~\cite{openai2024gpt4technicalreport,zhao2025surveylargelanguagemodels}. However, their limited parameterization and training data can still lead to hallucinated or stale outputs~\cite{ma2025comprehensivesurveyvectordatabase,bang2023multitaskmultilingualmultimodalevaluation,xu2025hallucinationinevitableinnatelimitation,Huang_2025}. Retrieval-Augmented Generation (RAG) mitigates this limitation by retrieving external knowledge at inference time, avoiding costly fine-tuning while improving factual grounding~\cite{lewis2020retrieval,gao2024retrievalaugmentedgenerationlargelanguage,borgeaud2022improving,izacard2022distillingknowledgereaderretriever,khandelwal2020generalizationmemorizationnearestneighbor,shi2023replugretrievalaugmentedblackboxlanguage}.

At the core of RAG pipelines are vector databases (VectorDBs)~\cite{10.1145/3448016.3457550,guo2022manucloudnativevector}, which store large-scale embeddings and support Approximate Nearest Neighbor Search (ANNS)~\cite{6619223,5432202,6809191,pmlr-v119-guo20h,10.5555/2976040.2976144}. Among ANNS algorithms, Hierarchical Navigable Small World (HNSW) graphs~\cite{malkov2018efficientrobustapproximatenearest} are widely adopted for their strong recall-latency trade-off and scalability~\cite{AUMULLER2020101374}.

\begin{figure}[t!]
  \centering
  \includegraphics[width=1.0\linewidth]{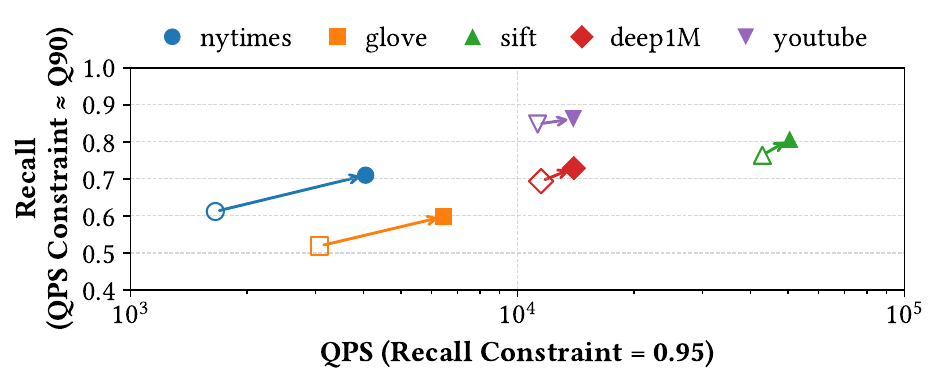}
  \caption{Performance differences between default and optimal HNSW hyperparameter configurations across datasets using Faiss~\cite{8733051}. Filled markers indicate optimal configurations, while outlined markers indicate default configurations. See Section~\ref{sec:setup} for details.}
  \label{fig:diff_default_optimal_hp}
\end{figure}

Despite its popularity, deploying HNSW effectively remains difficult~\cite{malkov2018efficientrobustapproximatenearest,autotuningtheconstructionparameters}. HNSW performance is governed by three core hyperparameters: 
\reviewerA{\M{}, the maximum number of connections allowed per node; \efc{}, the construction-time candidate-list size used during insertion; and \efs{}, the query-time candidate-list size used during search.} 
These parameters jointly affect recall, throughput, build time, and index size. \reviewerB{As shown in Figure~\ref{fig:diff_default_optimal_hp}, fixed default HNSW configurations often do not coincide with the dataset-specific feasible frontier or constraint-aware optimum. Thus, even under identical recall or QPS targets, tuned configurations can avoid constraint violations or recover performance left on the table by defaults, achieving up to 144\% higher throughput or 16\% higher recall.}


HNSW tuning is challenging because these parameters interact nonlinearly and often trade off competing objectives~\cite{yu2020hyperparameteroptimizationreviewalgorithms}. For example, increasing \efc{} can improve graph quality and recall, but may also increase build cost, index size, and query traversal overhead. Existing tuning methods, including Grid/Random Search~\cite{liashchynskyi2019gridsearchrandomsearch,10.5555/2188385.2188395}, Optuna~\cite{10.1145/3292500.3330701}, NSGA-II~\cite{996017}, ECI~\cite{gardner2014bayesian}, and VDTuner~\cite{yang2024vdtunerautomatedperformancetuning}, largely treat HNSW as an opaque objective. More broadly, generic surrogate-based and Bayesian optimizers can be sample-inefficient in discrete, noisy, and narrow-feasibility search spaces~\cite{7352306,10.5555/2999325.2999464,10.5555/3013558.3013569}, while Pareto-oriented methods~\cite{1599245} do not directly target the single constraint-satisfying optimum required here. These methods therefore spend substantial tuning budget exploring expensive configurations without directly exploiting the structural regularities induced by HNSW construction and search.

We propose \frameworkname{}, a constraint-aware, API-level black-box but HNSW-structure-aware tuning framework. \frameworkname{} does not inspect internal graph statistics; it uses only standard build/query APIs and measured validation signals. However, it exploits public HNSW mechanisms to decompose tuning into structured subproblems: guided search over construction parameters \((\M{}, \efc{})\), binary boundary search over \efs{}, and resource-feasibility filtering before full index construction. Resource estimates are used only to avoid unnecessary builds; configurations that survive this filter are still built, measured, and validated before selection. This design leverages monotonic, dominant unimodal, and separable relationships in the HNSW configuration space to reduce unnecessary evaluations while selecting configurations validated against user-specified performance and resource constraints. We intentionally scope these structural claims to standard HNSW implementations; extending the methodology to other ANN families requires identifying analogous search-time, construction-time, and resource regularities.

We evaluate \frameworkname{} on Faiss~\cite{8733051}, Hnswlib~\cite{malkov2018efficientrobustapproximatenearest}, and Milvus~\cite{10.1145/3448016.3457550} across five datasets. Compared with strong baselines, \frameworkname{} achieves up to 45\% higher throughput or 11\% higher recall, and converges up to $44\times$ faster, while satisfying accuracy, throughput, build-time, index-size, and tuning-budget constraints.

\noindent {Our key contributions are:}
\begin{itemize}[leftmargin=*, labelsep=0.5em, topsep=0pt, itemsep=0pt]
\item We formulate constraint-aware HNSW hyperparameter tuning under performance, resource, and tuning-budget constraints.
\item We characterize HNSW-specific structural regularities---monotone feasibility boundaries induced by \efs{}, dominant unimodal construction-parameter trends, and separable resource dependencies.
\item We introduce \frameworkname{}, an API-level black-box but HNSW-structure-aware tuner that uses hierarchical search and measured validation to identify constraint-satisfying configurations efficiently.
\item We develop analytic resource models for build time and index size, enabling feasibility-aware pruning before full index construction.
\item We validate \frameworkname{} across datasets and HNSW-based VectorDB backends, showing substantial performance and convergence gains over strong baselines.
\end{itemize}

%% file: sections/background.tex
\section{\reviewerA{Background}}
\label{sec:background}

\begin{figure}[t]
    \centering
    \includegraphics[width=.95\linewidth]{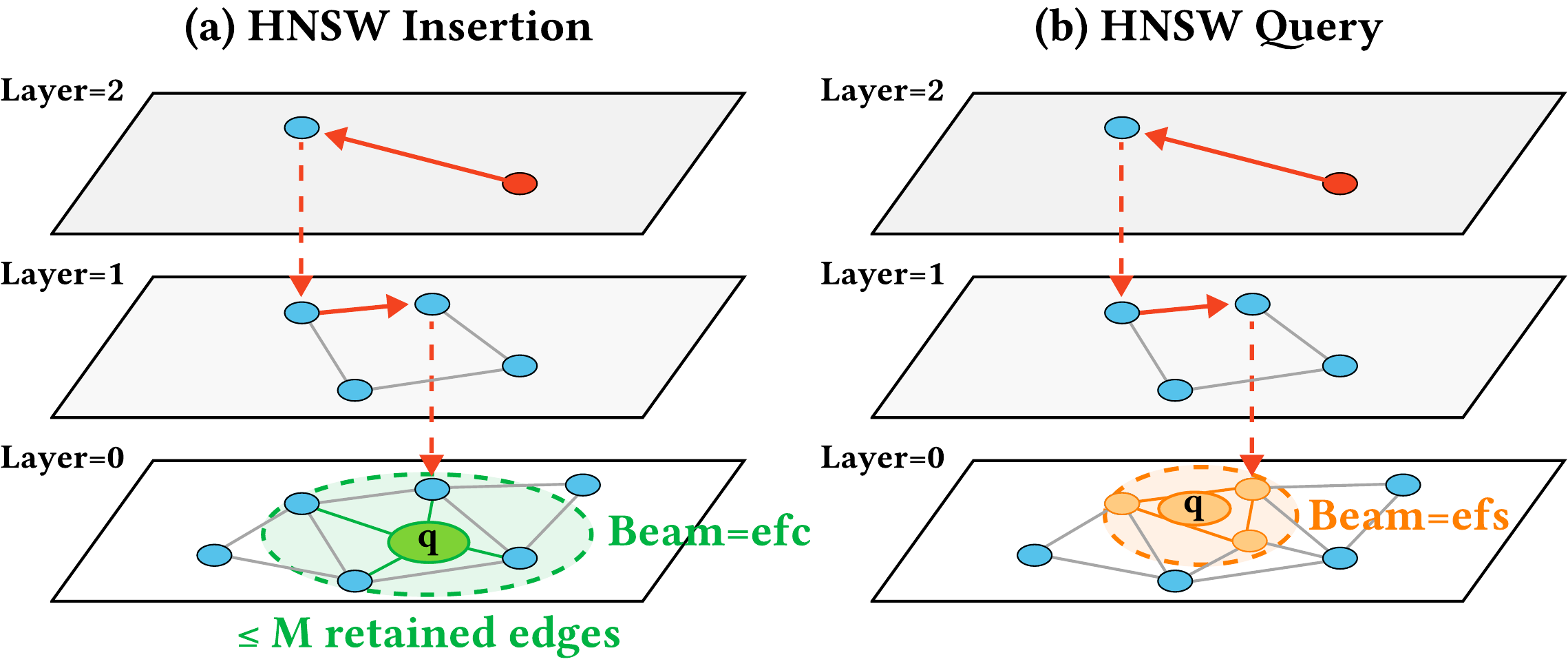}
    \caption{
    \reviewerA{HNSW insertion and query processing.
    (a) Insertion uses \efc{} to search candidate neighbors and \M{} to bound the close/diverse edges retained by \textsc{Select-Neighbors}.
    (b) Query processing uses \efs{} to control base-layer search after top-down greedy traversal.}
    }
    \label{fig:hnsw_background}
\end{figure}

\reviewerA{
Hierarchical Navigable Small World (HNSW)~\cite{malkov2018efficientrobustapproximatenearest} is a hierarchical graph-based index for Approximate Nearest Neighbor Search (ANNS). Each data point is represented as a node, and proximity relationships are represented as graph edges. HNSW exposes three core tuning parameters used throughout this paper: \M{}, the degree budget that controls how many neighbors a node can retain; \efc{}, the construction-time candidate-list size used when inserting nodes; and \efs{}, the search-time candidate-list size used when answering queries. The construction parameters \M{} and \efc{} determine the graph structure, whereas \efs{} controls query-time traversal over a fixed graph.

Figure~\ref{fig:hnsw_background} illustrates the layered structure and the two main HNSW procedures. Each node is assigned a maximum layer by randomized level generation and appears in every layer from Layer~0 up to that level. Since upper layers contain progressively fewer nodes, they provide sparse long-range routing links, while Layer~0 contains all data points and supports fine-grained search. Both insertion and query processing therefore follow a coarse-to-fine pattern: HNSW first performs greedy top-down traversal through upper layers and then conducts a broader search near the Layer 0.

During insertion (Figure~\ref{fig:hnsw_background}(a)), a new point $q$ first uses greedy descent with $\text{Beam}=1$ to locate an entry point near $q$. From the assigned layer of $q$ down to Layer~0, HNSW runs \textsc{Search-Layer} with candidate-list size \efc{} to collect candidate neighbors. These candidates are then passed to \textsc{Select-Neighbors}, which retains close and diverse neighbors up to the layer-specific degree budget induced by \M{}; fewer than \M{} edges may be retained when candidates are redundant. During query processing (Figure~\ref{fig:hnsw_background}(b)), the query vector $q$ is not inserted into the graph. After the same greedy top-down traversal, HNSW runs \textsc{Search-Layer} at Layer~0 with candidate-list size \efs{} and returns the top-$k$ candidates from the final working set. Thus, increasing \efs{} broadens query-time search on a fixed graph, while changing \M{} or \efc{} changes the graph and requires rebuilding the index. 
Appendix~A provides formal pseudocode for these procedures.
}

%% file: sections/hyperparameters_performance.tex
\section{Decoding HNSW Hyperparameter Trade-offs}
\label{sec:hyperparameters_performance}


This section presents an in-depth analysis of the three key HNSW 
hyperparameters (i.e., \M{}, \efc{}, and \efs{}) and their 
non-trivial interactions in shaping search 
performance~\cite{malkov2018efficientrobustapproximatenearest}. 
Because these parameters are tightly coupled, changing one often 
requires compensatory adjustments in others; understanding these 
trade-offs is essential for principled tuning and motivates our 
\frameworkname{} framework. We focus here on Recall and query 
throughput (QPS), and show that the strict monotonicity of \efs{} 
with respect to recall and latency induces a total order over 
configurations, enabling binary feasibility search for constraint 
satisfaction rather than heuristic exploration. Resource-usage 
metrics (index build time and size) are analyzed in 
Section~\ref{sec:hyperparameters_resource} to inform our constraint 
surrogate models. Unless otherwise specified, figures referenced 
in this section are generated using Faiss~\cite{8733051} on the 
nytimes dataset~\cite{AUMULLER2020101374} (see 
Section~\ref{sec:setup} for details on the experimental setup).
\reviewerB{
The qualitative patterns described below were consistent across 
query parallelism levels (1, 4, 16, and 64 threads); the figures 
in this section use 64 threads, and we focus on these structural 
trends rather than absolute QPS values.
}

\begin{figure}
  \centering
  \begin{subfigure}[b]{0.48\linewidth}
    \includegraphics[width=\linewidth]{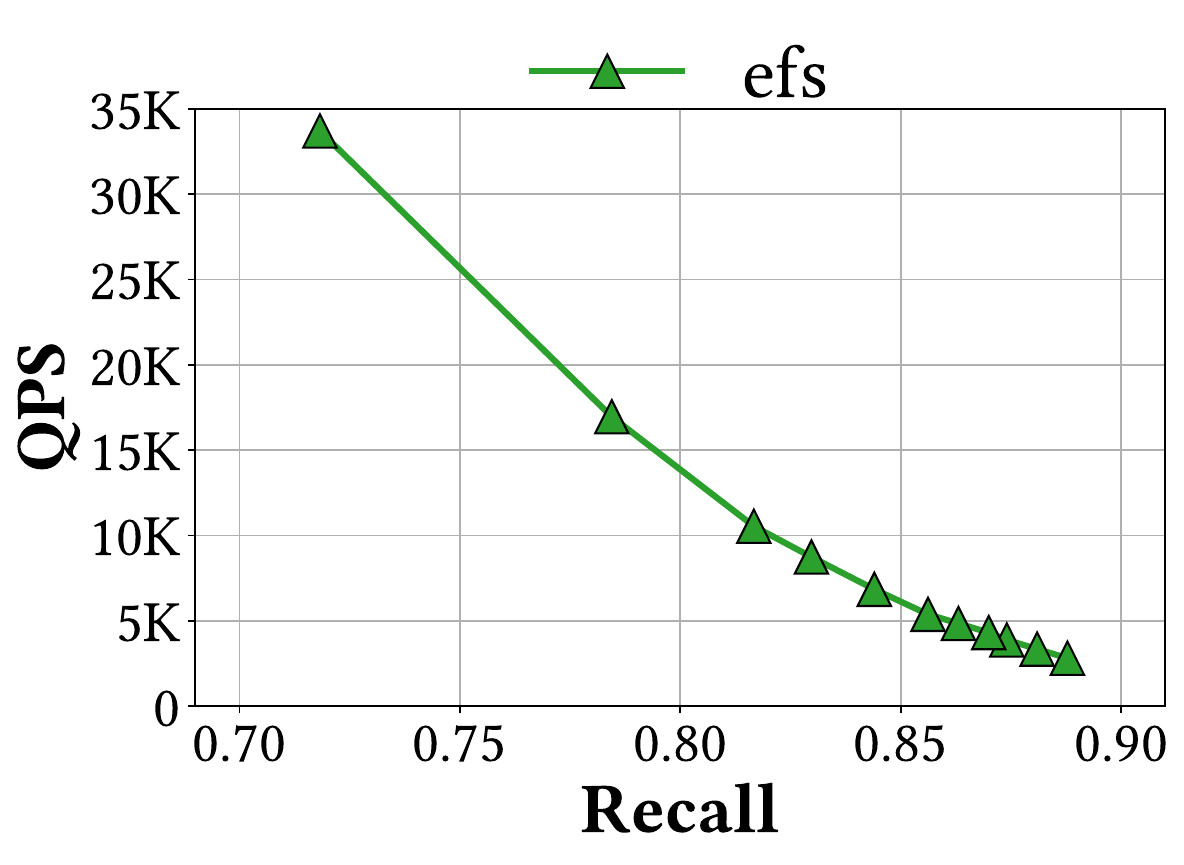}
    \vspace{-1.5em}
    \caption{}
    \label{fig:sub_a}
  \end{subfigure}\hfill
  \begin{subfigure}[b]{0.48\linewidth}
    \includegraphics[width=\linewidth]{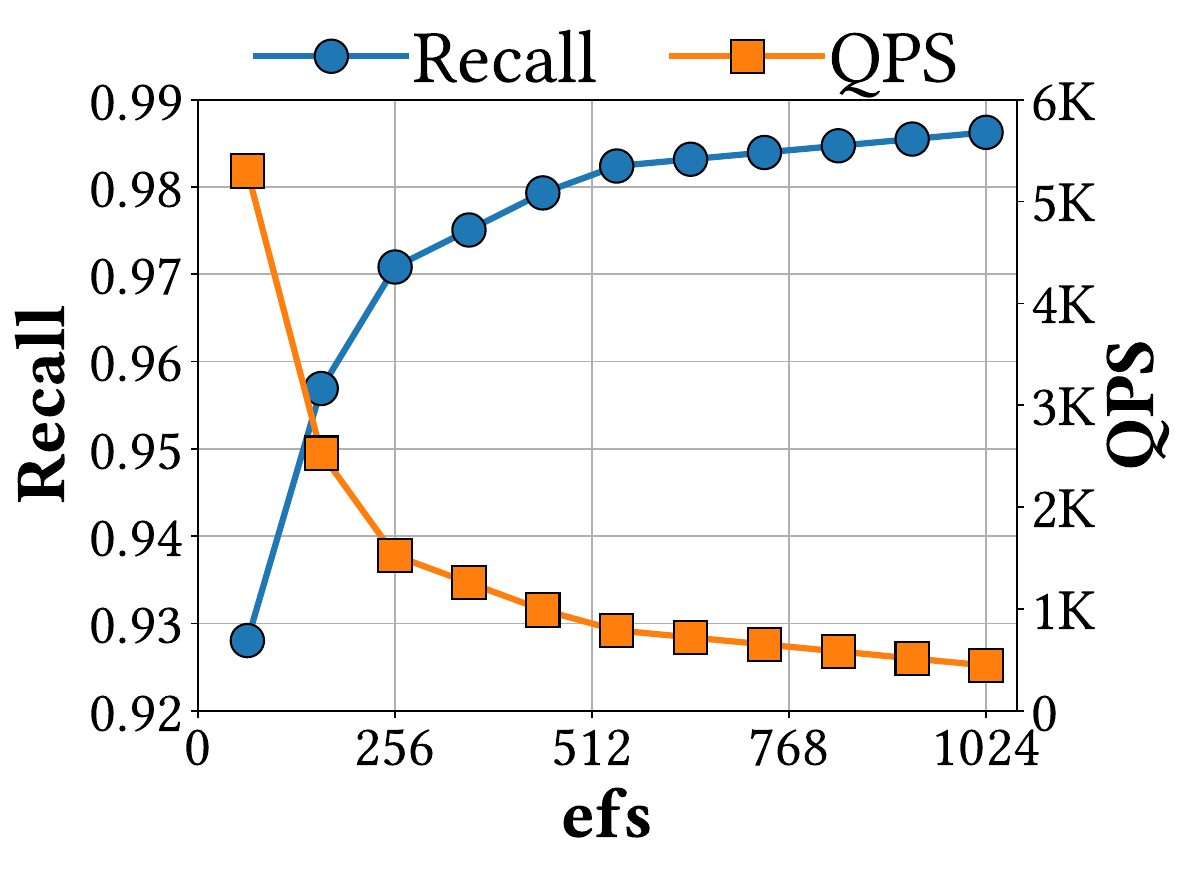}
    \vspace{-1.5em}
    \caption{}
    \label{fig:sub_b}
  \end{subfigure}
  \vspace{-1em}
  \caption{Effect of \efs{} on search performance. (a) Recall-QPS trade-off as \efs{} increases. (b) Individual performance curves showing recall saturation and QPS decline with increasing \efs{}.}
  \label{fig:efs_recall_qps}
\end{figure}

\subsection{Effect of \efs{} on Search Accuracy and Speed}
\label{sec:effect_of_efs}

\gun{
The hyperparameter \efs{} controls the size of the candidate set maintained and expanded during query-time traversal at the base layer. Increasing \efs{} broadens the search scope, thereby improving recall, but also incurs higher computational costs due to additional distance evaluations, resulting in lower query throughput (QPS). This monotonic accuracy--speed trade-off is illustrated in Figure~\ref{fig:efs_recall_qps}(a).
}

\gun{
Beyond this monotonic trade-off, \efs{} exhibits diminishing returns in the high-\efs{} regime. As shown in Figure~\ref{fig:efs_recall_qps}(b), recall improvements rapidly saturate once most informative candidates near the query have already been explored, making additional candidates largely distant or redundant. QPS degradation also flattens in this regime because the number of candidates that can be meaningfully expanded is constrained by the graph structure and by the finite set of nearby nodes that can contribute to the top-$k$ results. Thus, while the nominal candidate budget increases, the marginal performance change from additional search effort decreases.
}

\gun{
These saturation behaviors stem from the greedy search mechanism employed at layer~0 in HNSW. As the search converges toward a local optimum, discovering substantially better candidates becomes increasingly difficult, causing the marginal utility of additional search effort---measured as performance gain per unit cost---to drop sharply.
Crucially, this structure induces a monotone feasibility ordering over \efs{}: larger \efs{} improves recall but increases query cost. This property enables binary feasibility search for constraint satisfaction, rather than costly heuristic exploration, and explains why blindly increasing \efs{} is inefficient for balancing accuracy and latency.
}

\newlength{\subfigH}
\setlength{\subfigH}{0.135\textheight}
\subsection{Interaction between Construction Parameters and \efs{}}
\label{sec:construction_efs}

\begin{figure}
  \centering
  \begin{subfigure}[b]{0.48\linewidth}
    \includegraphics[height=\subfigH,keepaspectratio]{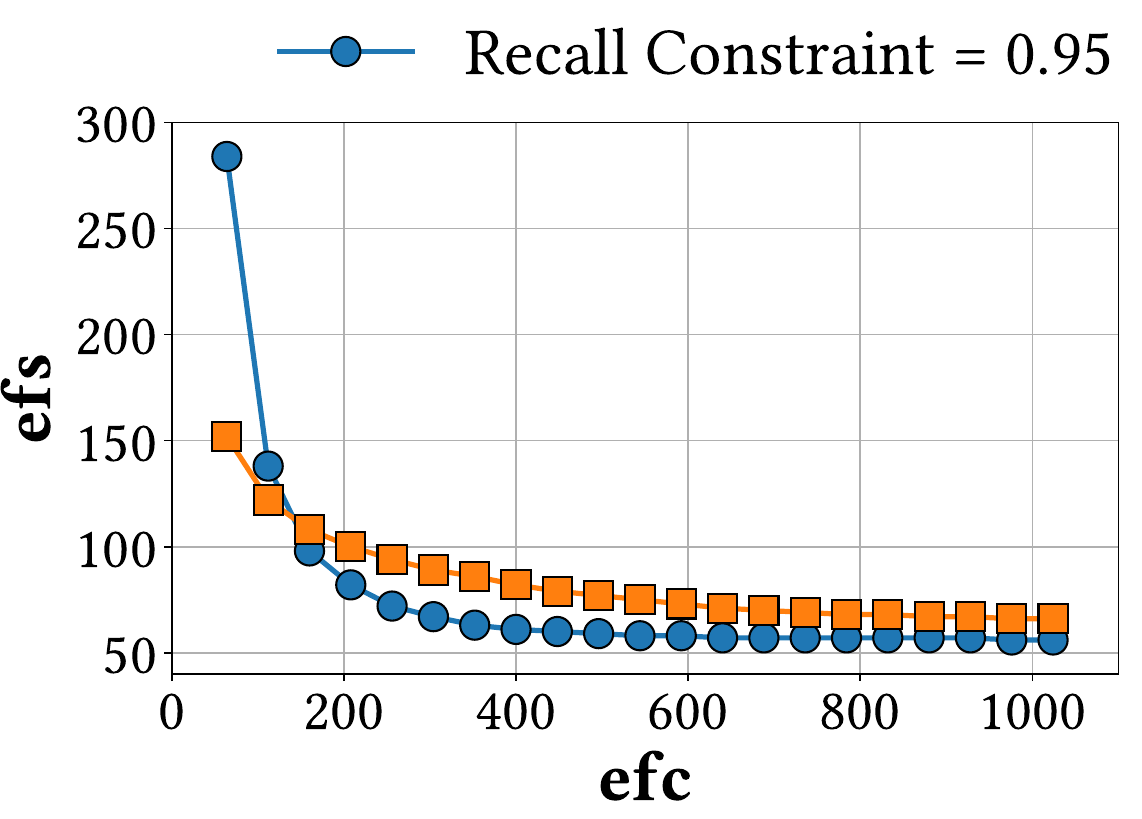}
    \vspace{-1.5em}
    \caption{}
    \label{fig:efs_efc_boundary}
  \end{subfigure}\hfill
  \begin{subfigure}[b]{0.48\linewidth}
    \includegraphics[height=\subfigH,keepaspectratio]{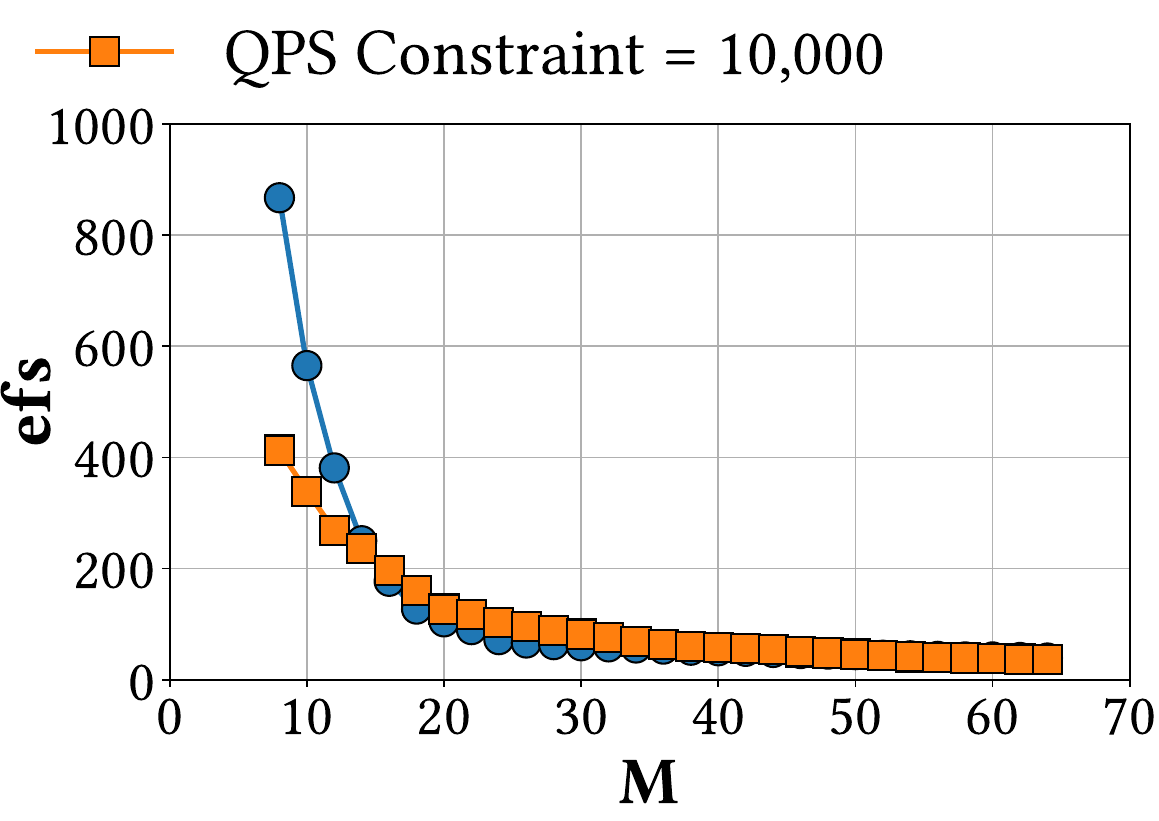}
    \vspace{-1.5em}
    \caption{}
    \label{fig:efs_m_boundary}
  \end{subfigure}
  \vspace{-1em}
  \caption{Hyperparameter interactions under performance constraints.
  (a) Increasing \efc{} shifts the Recall-feasibility boundary toward
  smaller \efs{}. (b) Increasing \M{} shifts the QPS-feasibility boundary
  toward smaller \efs{}. Both exhibit steep initial movement followed by
  saturation.}
  \vspace{-1em}
  \label{fig:efs_efc_M_Performance}
\end{figure}

The construction-time parameters \efc{} and \M{} shape the HNSW graph through complementary mechanisms. The parameter \efc{} controls the breadth of candidate search during insertion: a larger \efc{} exposes \textsc{Select-Neighbors} to a richer candidate pool, increasing the likelihood that high-quality, diverse edges survive pruning. Because \textsc{Select-Neighbors} often retains fewer than \M{} edges due to its diversification criterion, larger \efc{} can improve realized connectivity even when the degree cap is not the binding constraint. \M{}, in contrast, directly caps per-node degree. Thus, \efc{} shapes \emph{which} candidate edges are available, while \M{} bounds \emph{how many} are retained.

To characterize the coupling with query-time search, we use the \emph{constraint-boundary} \efs{}: the minimum \efs{} that reaches the target Recall under a Recall constraint, or the maximum \efs{} that keeps throughput above the target under a QPS constraint. Figure~\ref{fig:efs_efc_M_Performance}(a) shows that increasing \efc{} shifts this boundary downward under a Recall constraint, because a richer candidate pool improves graph quality and fewer query-time candidates suffice. Figure~\ref{fig:efs_efc_M_Performance}(b) shows the similar shift for \M{} under a QPS constraint: larger degree caps increase work per expanded node, reducing the maximum feasible \efs{}. The boundary moves in the same direction under the opposite constraint for both parameters, because stronger realized connectivity simultaneously improves Recall (lowering minimum-feasible \efs{}) and raises traversal cost (lowering maximum-feasible \efs{}).

The steep-to-gradual shape in both panels reflects a common saturation mechanism in construction-time graph improvement. For \efc{}, the benefit saturates because \textsc{Select-Neighbors} increasingly rejects redundant candidates once the candidate pool is sufficiently rich. For \M{}, realized connectivity can grow sublinearly and saturate below the nominal cap because the same diversification rule limits which additional neighbors are retained. Thus, \efc{} and \M{} can both reduce the boundary \efs{}, but their marginal benefit diminishes once the graph becomes sufficiently navigable.

These observations have two implications for tuning. First, stronger realized connectivity typically shifts the constraint-boundary \efs{} downward, allowing \frameworkname{} to narrow the \efs{} search range using previous evaluations (Section~\ref{sec:efs_search}). Second, when graph-quality gains saturate while traversal overhead continues to grow, the constrained objective moves from a quality-improving regime to an overhead-dominated regime. This crossover underlies the dominant unimodal trends analyzed in Sections~\ref{sec:M_efc_perf} and~\ref{sec:performance_across_M}.

\begin{figure}[t]
  \centering
  \includegraphics[width=1.0\linewidth]{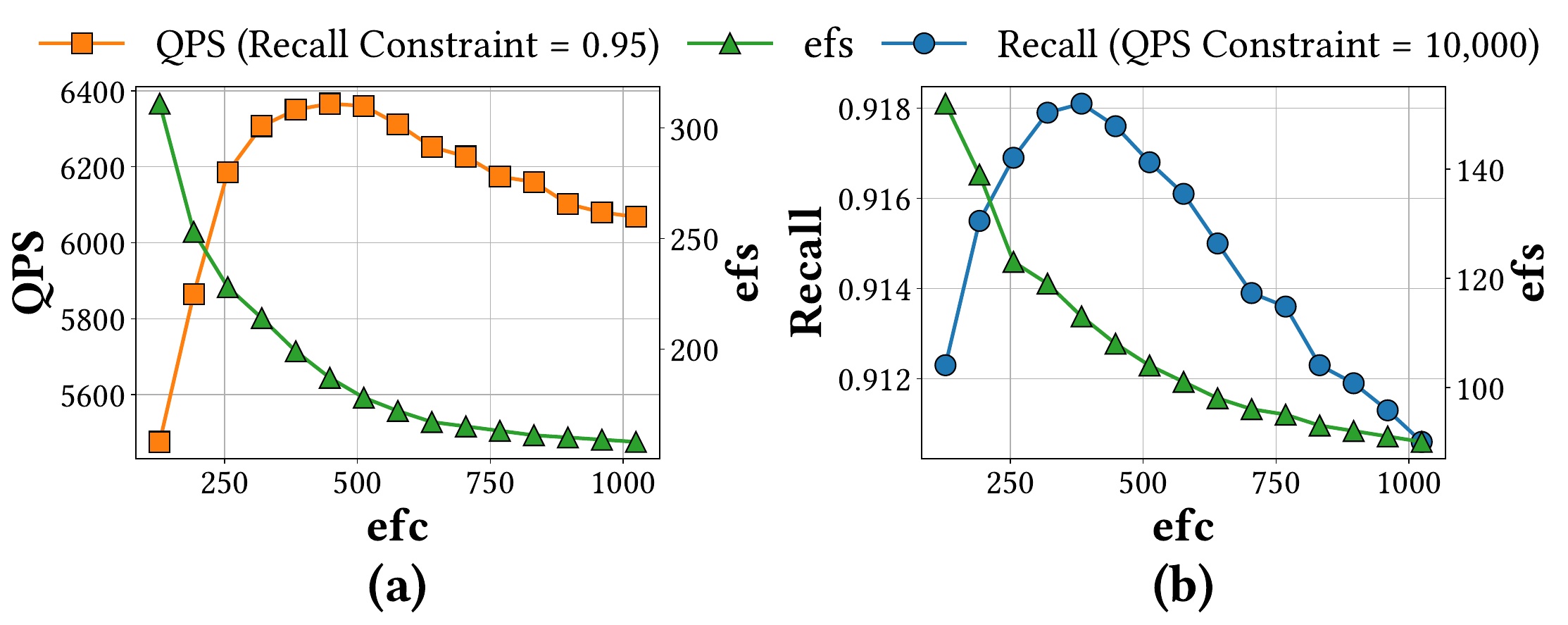}
  \vspace{-2em}
  \caption{Dominant unimodal performance patterns with varying
  \efc{} under fixed \M{}. (a) Under a Recall constraint, QPS
  first improves as increasing \efc{} reduces the constraint-boundary
  \efs{}, then declines as traversal overhead dominates.
  (b) Recall under a QPS constraint follows the same pattern.}
  \label{fig:M_efc_twin_efs}
\end{figure}

\subsection{\reviewerB{Unimodality over \efc{} under Fixed \M{}}}
\label{sec:M_efc_perf}

We fix \M{} and study how \efc{} affects the constrained objective after \efs{} has been chosen at the feasibility boundary. For each \efc{}, \frameworkname{} selects the boundary \efs{} under the active performance constraint: the smallest \efs{} satisfying the Recall target, or the largest \efs{} satisfying the QPS target. Let $f_M(efc)$ denote the resulting objective value, namely QPS under a Recall constraint and Recall under a QPS constraint. We call $f_M$ unimodal if it is non-decreasing up to some peak index and non-increasing thereafter.

Figure~\ref{fig:M_efc_twin_efs} shows this dominant unimodal behavior under both constraints. When \efc{} is small, the graph is poorly connected, so the boundary \efs{} is unfavorable: Recall constraints require large \efs{}, while QPS constraints force small \efs{} and sacrifice Recall. Increasing \efc{} enlarges the construction-time candidate pool and improves graph navigability, allowing the constraint to be met more efficiently and raising the objective. Beyond the peak, additional candidates become increasingly redundant under \textsc{Select-Neighbors} and the fixed degree cap \M{}, while traversal overhead continues to grow. The objective therefore shifts from a quality-improving regime to an overhead-dominated regime.

\reviewerC{
\begin{lemma}[Sufficient condition for unimodality over \efc{}]
\label{lem:efc_unimodal}
Fix \M{} and an active constraint $C\in\{\mathrm{Recall},\mathrm{QPS}\}$. On the ordered grid $efc_1<\cdots<efc_n$, write the marginal change as
$\Delta f_M(efc_i)=\Delta B_M(efc_i)-\Delta T_M(efc_i)$, where $\Delta B_M$ is the marginal graph-quality benefit and $\Delta T_M$ is the residual traversal/search cost. Any reduction in expanded nodes from better graph navigability is counted in $\Delta B_M$, not as negative cost. If
\begin{enumerate}[label=(\roman*),leftmargin=1.5em,topsep=0.1em,itemsep=0em]
\item $\Delta T_M(efc_i)\ge0$ for all $i$, and
\item once $\Delta B_M(efc_i)\le\Delta T_M(efc_i)$, larger \efc{} values do not reveal a delayed high-utility connectivity regime that restores $\Delta B_M>\Delta T_M$,
\end{enumerate}
then $f_M(efc)$ is unimodal.
\end{lemma}
\noindent\emph{Proof.}
Due to space limitations, we defer the proof sketch to Appendix~B.2.
}

Lemma~\ref{lem:efc_unimodal} is sufficient, not universal. It captures the standard HNSW regime targeted by \frameworkname{}: bounded degree, diversification-based pruning, and no delayed high-utility connectivity regime. When condition~(ii) fails, the constrained objective may become bimodal or multi-modal; Section~\ref{sec:structural_robustness} presents such a counterexample. When the conditions hold, $f_M$ can be searched efficiently using the ternary-style procedure in Section~\ref{sec:solution}.

\subsection{\reviewerB{Outer-Objective Unimodality over \M{}}}
\label{sec:performance_across_M}

We next analyze how the best achievable objective changes as \M{} varies. 
For each \M{}, \frameworkname{} chooses the best construction effort, denoted
$efc_M^* \in \arg\max_{efc} F_C(M,efc)$, and the boundary \efs{} satisfying
the active performance constraint. Let $g(M)=\max_{efc}F_C(M,efc)$ denote the
resulting outer objective: QPS under a Recall constraint and Recall under a
QPS constraint. We call $g(M)$ unimodal if it is non-decreasing up to some
peak index and non-increasing thereafter.

\begin{figure}[t!]
  \centering
  \begin{subfigure}[b]{0.48\linewidth}
    \includegraphics[width=\linewidth]{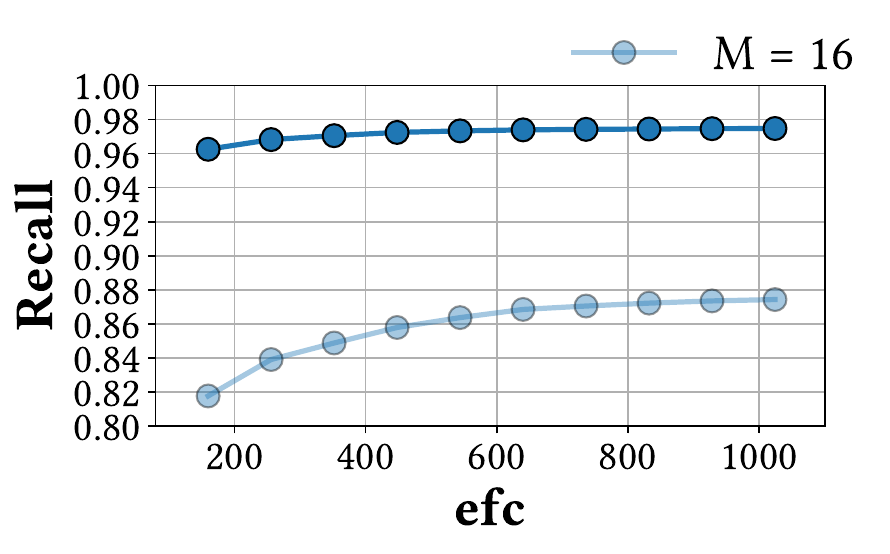}
    \vspace{-1.5em}
    \caption{}
    \label{fig:graph_quality}
  \end{subfigure}\hfill
  \begin{subfigure}[b]{0.48\linewidth}
    \includegraphics[width=\linewidth]{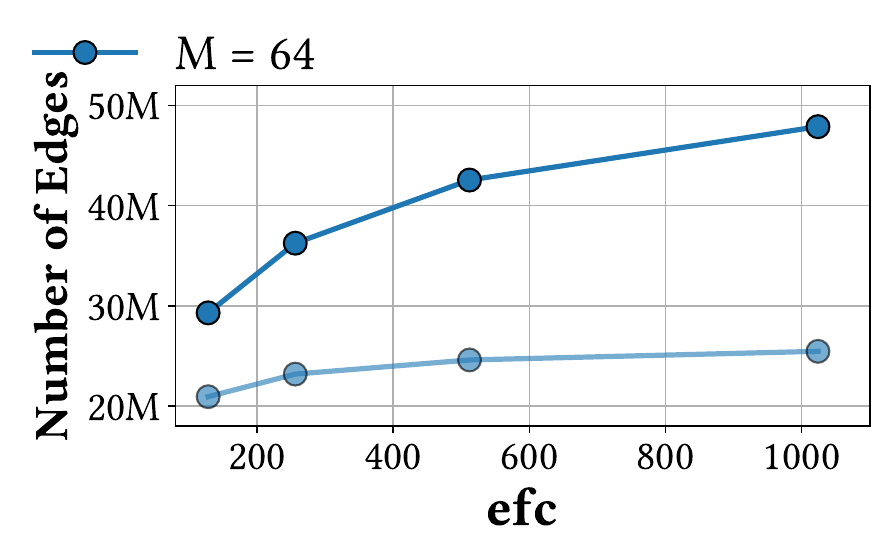}
    \vspace{-1.5em}
    \caption{}
    \label{fig:edge_growth}
  \end{subfigure}
  \vspace{-1em}
  \caption{Graph quality and realized edge growth across different \M{} values.
  (a) Recall as a function of \efc{}, showing that larger \M{} values reach high graph quality with smaller \efc{}.
  (b) Number of edges as a function of \efc{}, showing saturation in realized connectivity under fixed degree caps.}
  \label{fig:graph_growth}
\end{figure}

Figure~\ref{fig:graph_growth} explains why $efc_M^*$ generally decreases as
\M{} increases. When \M{} is small, each node can retain only a limited number
of neighbors, so larger \efc{} is useful for exposing high-quality and diverse
candidate edges. When \M{} is larger, the graph can already retain more useful
routing alternatives, so similar navigability can be achieved with a smaller
construction-time candidate pool. Figure~\ref{fig:graph_growth}(b) confirms
this substitution effect: realized edges saturate quickly under tight degree
caps, while larger \M{} allows the graph to keep growing as larger \efc{}
exposes more useful candidates. Thus, \M{} and \efc{} act as partial substitutes
for graph quality, explaining the dominant decreasing trend of $efc_M^*$.

\begin{figure}[t!]
  \centering
  \includegraphics[width=1.0\linewidth]{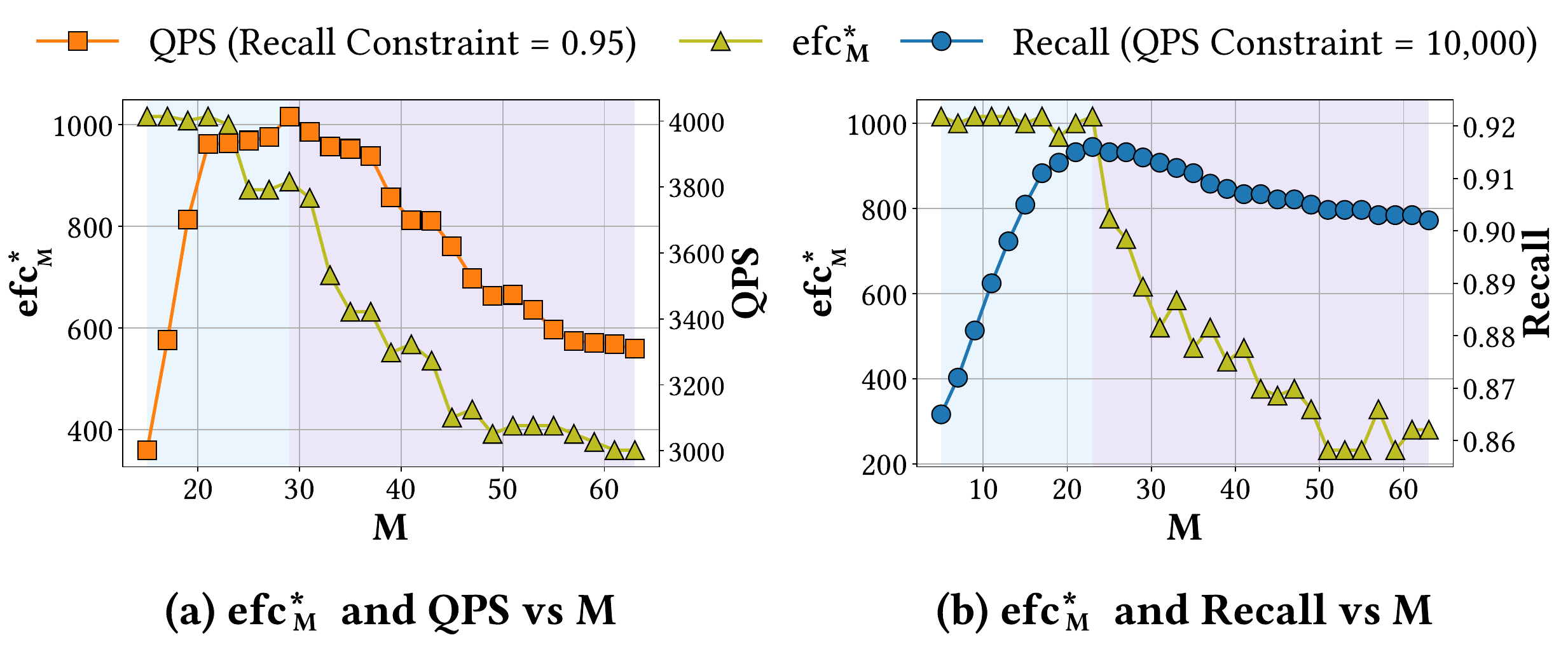}
  \vspace{-2em}
  \caption{Optimized \efc{} and objective performance across varying \M{} values.
  (a) $efc_M^*$ and resulting QPS versus \M{} under a Recall constraint.
  (b) $efc_M^*$ and resulting Recall versus \M{} under a QPS constraint.}
  \label{fig:M_outer_envelope}
\end{figure}

Figure~\ref{fig:M_outer_envelope} shows the resulting outer objective.
Under a Recall constraint, small \M{} produces poor connectivity, so the
Recall target requires a large boundary \efs{} and suppresses QPS. Increasing
\M{} initially improves graph navigability and reduces the required \efs{}, so
QPS improves. At larger \M{}, however, additional degree capacity provides
diminishing graph-quality benefit while increasing neighbor processing per
expanded node, shifting the objective to an overhead-dominated regime. Under a
QPS constraint, the same trade-off appears in the opposite objective: increasing
\M{} initially improves Recall, but excessive degree forces a smaller boundary
\efs{} to preserve throughput, eventually reducing Recall. Both cases show a
dominant unimodal trend in $g(M)$.

\reviewerC{
\begin{lemma}[Sufficient condition for outer-objective unimodality over \M{}]
\label{lem:m_unimodal}
Fix an active constraint $C\in\{\mathrm{Recall},\mathrm{QPS}\}$. On the ordered
grid $M_1<\cdots<M_n$, write the marginal change as
$\Delta g(M_i)=\Delta B_{\mathrm{out}}(M_i)-\Delta T_{\mathrm{out}}(M_i)$,
where $\Delta B_{\mathrm{out}}$ is the marginal graph-quality benefit from
increasing \M{} with \efc{} re-optimized at each \M{}, and
$\Delta T_{\mathrm{out}}$ is the residual traversal/search cost. Any reduction
in expanded nodes from better graph navigability is counted in
$\Delta B_{\mathrm{out}}$, not as negative cost. If
\begin{enumerate}[label=(\roman*),leftmargin=1.5em,topsep=0.1em,itemsep=0em]
\item $\Delta T_{\mathrm{out}}(M_i)\ge0$ for all $i$,
\item once $\Delta B_{\mathrm{out}}(M_i)\le\Delta T_{\mathrm{out}}(M_i)$,
larger \M{} values do not reveal a delayed high-utility connectivity regime
that restores $\Delta B_{\mathrm{out}}>\Delta T_{\mathrm{out}}$, and
\item re-optimizing \efc{} at larger \M{} does not create a new high-performing
branch, i.e., \M{} and \efc{} remain partial substitutes for graph quality,
\end{enumerate}
then $g(M)$ is unimodal over the ordered \M{} grid.
\end{lemma}
\noindent\emph{Proof.}
Due to space limitations, we defer the proof sketch to Appendix~B.2.
}

Lemma~\ref{lem:m_unimodal} is sufficient, not universal. In particular,
Lemma~\ref{lem:efc_unimodal} alone does not imply unimodality of the outer
envelope $g(M)$; condition~(iii) rules out recoveries caused by re-optimizing
\efc{} into a delayed high-utility regime at larger \M{}. If these conditions
fail, $g(M)$ may become bimodal or multi-modal, as discussed in
Section~\ref{sec:structural_robustness}. When they hold, the outer search over
\M{} can be performed efficiently using the ternary-style procedure in
Section~\ref{sec:solution}.

\subsection{Structural Implications}
\jae{Taken together, these observations reveal that the HNSW hyperparameter space is not an unstructured three-dimensional landscape. Instead, it forms a partially ordered search space with distinct structural properties:}

\begin{itemize}[leftmargin=*, labelsep=0.5em, topsep=0pt, itemsep=0pt]
\item{\textit{\textbf{efs}} induces a monotonic trade-off between recall and latency, enabling binary feasibility search.}
\item{\textbf{M} exhibits unimodal objective behavior under fixed constraints, enabling efficient ternary optimization.}
\item{\textit{\textbf{efc}} acts as a feasibility gate that constrains construction-time resources.}
\end{itemize}

\jae{These properties fundamentally distinguish HNSW tuning from generic black-box optimization problems.
They enable deterministic, low-sample optimization strategies that exploit algorithmic structure rather than relying on random exploration. While the precise curvature of these trade-offs varies across datasets and backends, the monotonic and unimodal structures consistently appear in standard HNSW implementations with bounded degree and greedy layer-wise search.}
\cgcg{
These structural properties also provide a lens for reasoning about
how the optimal construction effort responds to corpus and workload
updates, as discussed next.
}

\begin{figure}[t!]
  \centering
  \includegraphics[width=1.0\linewidth]{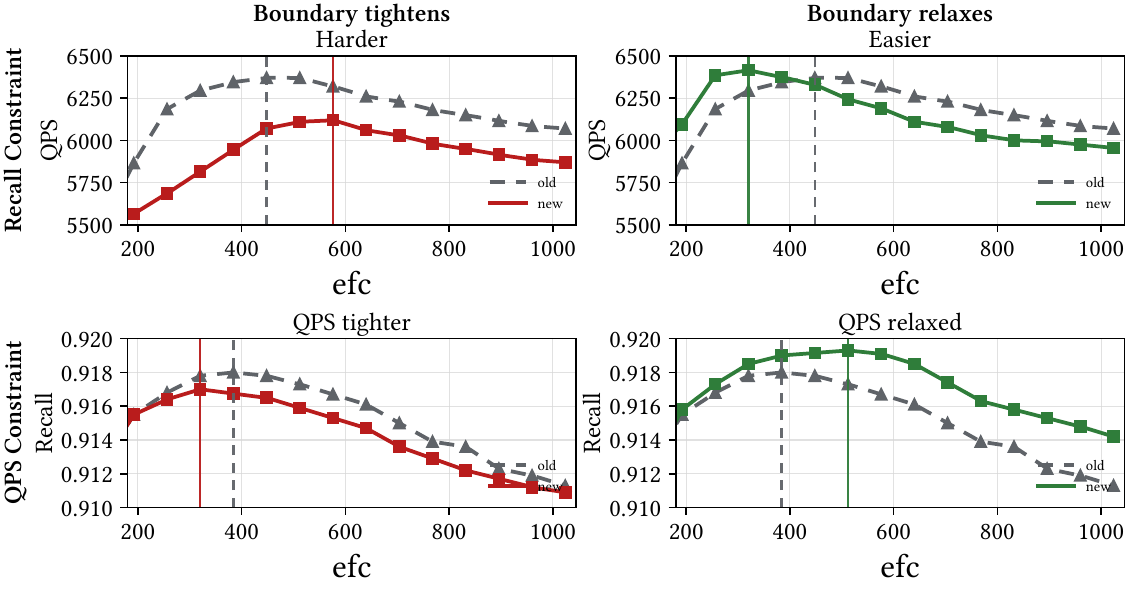}
  \vspace{-2em}
    \caption{
    \cgcg{
    Drift implications of \efs{}-boundary shifts. Under a Recall
    constraint, \starefcm{} moves right when the required \efs{}
    increases and left when it decreases. Under a QPS constraint,
    \starefcm{} moves left when the maximum allowed \efs{}
    decreases and right when it increases.
    }
    }
  \label{fig:drift_boundary_shift}
\end{figure}

\cgcg{
\subsection{Drift Implications of Feasibility Boundaries}
\label{sec:drift-implications}

The preceding analysis assumes a fixed corpus and query set. We now 
consider how the same trade-offs change when the corpus or query 
distribution changes. Under a Recall constraint, \starefcm{} marks the 
transition between a quality-limited region---where increasing \efc{} 
still substantially reduces the required \efs{}---and a cost-limited 
region, where graph-quality gains have saturated. A change that 
\emph{tightens} the Recall boundary $efs^*_{\mathrm{Recall}}(M,efc)$ 
(the minimum \efs{} satisfying the Recall target increases) expands the 
quality-limited region and moves $efc^*_M$ toward larger \efc{}; 
\emph{relaxing} the boundary shifts $efc^*_M$ in the opposite direction. 
Under a QPS constraint, tightening the boundary 
$efs^*_{\mathrm{QPS}}(M,efc)$ (the maximum feasible \efs{} decreases) 
makes high-\efc{} configurations cost-dominated earlier, so $efc^*_M$ 
moves leftward; relaxing it shifts $efc^*_M$ rightward.

Figure~\ref{fig:drift_boundary_shift} summarizes these directional 
effects. Insertions, deletions, and query-distribution shifts can each 
tighten or relax the active boundary depending on their net effect; the 
change type alone therefore does not determine the direction of movement. 
These directions should be understood as typical structural behaviors 
rather than universal guarantees, especially under near-duplicate 
insertions, targeted deletions, or strong query distribution shifts.
}

%% file: sections/hyperparameters_resource.tex
\section{\reviewerB{Relationship between Hyperparameters and HNSW Resources}}
\label{sec:hyperparameters_resource}

This section characterizes how the HNSW construction hyperparameters
$(M,\textit{efc})$ affect build time and index size. We separate the
algorithmic dependence induced by \textsc{Insert}, \textsc{Search-Layer}, and
\textsc{Select-Neighbors} from backend-specific constants, treating $N$ and
$d$ as fixed workload parameters. The goal is not to introduce a new objective,
but to identify structural resource dependencies that enable
feasibility-aware pruning. Unless otherwise specified, measurements use the
same experimental setup as Section~\ref{sec:hyperparameters_performance}.

\subsection{Index Construction Time}
\label{sec:build_time_modeling}


HNSW constructs the index by inserting points sequentially. For a dataset
$D=\{x_i\}_{i=1}^{N}$, the total build time decomposes over insertions:
\begin{equation}
T_{\mathrm{build}}(N,M,\textit{efc},d)
=
\sum_{i=1}^{N}T_{\mathrm{ins}}(i;M,\textit{efc},d).
\label{eq:build_sum}
\end{equation}
Each insertion descends from the top layer to layer $0$, performs \textsc{Search-Layer} to identify candidate neighbors at each visited layer, and then applies \textsc{Select-Neighbors} to finalize outgoing edges. Let $L_i$ denote the maximum layer of the $i$-th inserted point. The insertion cost can be decomposed as
\begin{equation}
T_{\mathrm{ins}}(i;M,\textit{efc},d)
=
\sum_{\ell=0}^{L_i}
T_{\mathrm{search}}(i,\ell;M,\textit{efc},d)
+
\sum_{\ell=0}^{L_i}
T_{\mathrm{select}}(i,\ell;M,d).
\label{eq:insert_decomp}
\end{equation}

\begin{figure}[t]
  \centering
  \begin{subfigure}[b]{0.5\linewidth}
    \includegraphics[width=\linewidth]{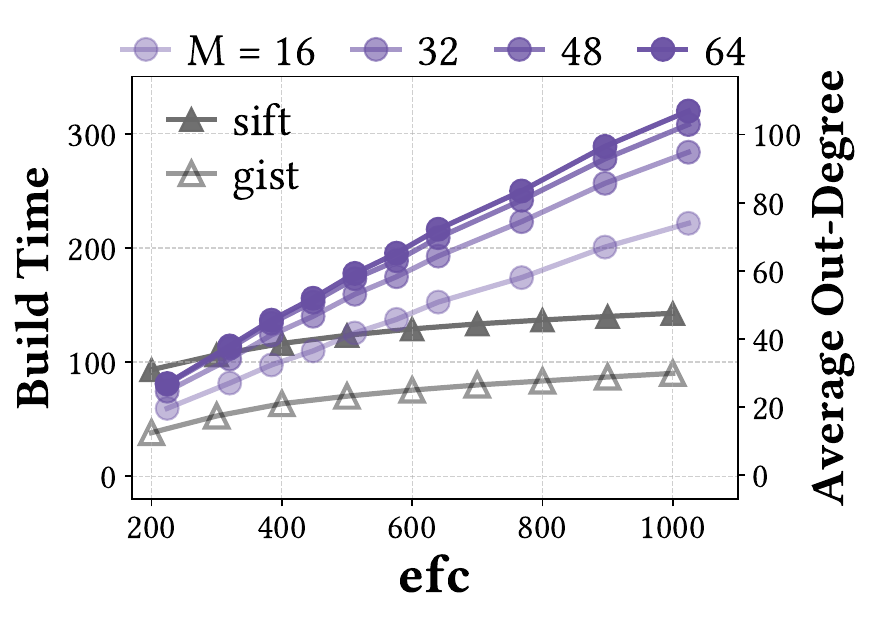}
    \vspace{-1.5em}
    \caption{}
    \label{fig:build_time_and_aod_by_efc}
  \end{subfigure}\hfill
  \begin{subfigure}[b]{0.5\linewidth}
    \includegraphics[width=\linewidth]{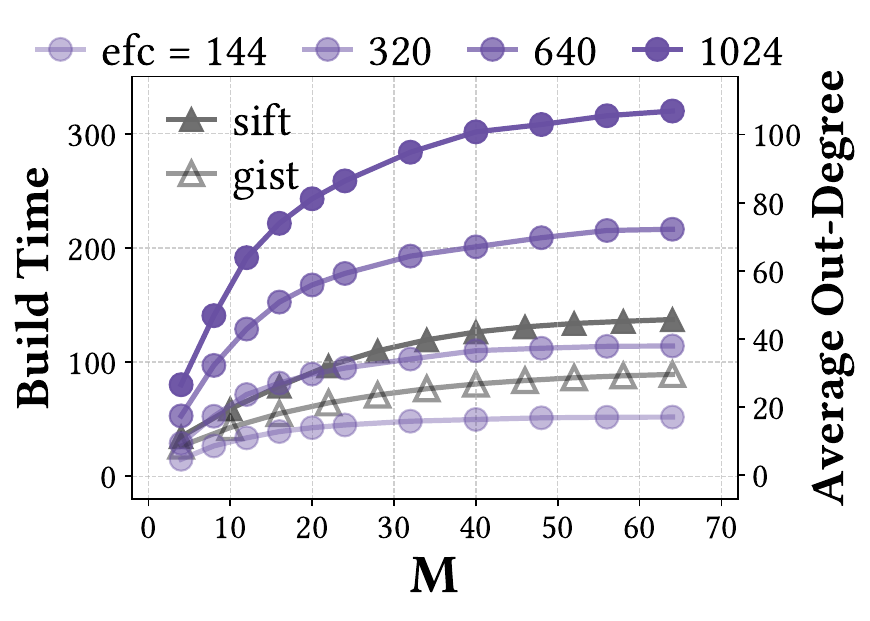}
    \vspace{-1.5em}
    \caption{}
    \label{fig:build_time_and_aod_by_m}
  \end{subfigure}
  \vspace{-2em}
  \caption{Build time and realized connectivity of HNSW as functions of the construction hyperparameters $\boldsymbol{(M,\textit{efc})}$. Index construction time (left y-axis) and average out-degree (AOD; right y-axis) are shown. (a) For fixed $\boldsymbol{M}$, both build time and AOD increase monotonically with $\textit{efc}$. (b) For fixed $\textit{efc}$, build time grows sublinearly in $M$, while AOD saturates as $\boldsymbol{M}$ increases.}
  \label{fig:build_time_by_m_efc}
\end{figure}

\textbf{Search-cost accounting.}
The dominant cost in SEARCH-LAYER comes from graph exploration: the
algorithm expands candidate nodes, inspects their adjacency lists, and
evaluates previously unseen candidates. Let
$X_{i,\ell}(M,\textit{efc})$ denote the number of expanded nodes when
inserting point $i$ at layer $\ell$, and let
$\bar{\delta}_{i,\ell}(M,\textit{efc})$ denote the average realized
out-degree inspected per expansion. The dominant search cost can then be
summarized as
\begin{equation}
T_{\mathrm{search}}(i,\ell;M,\textit{efc},d)
\approx
c_d\,
X_{i,\ell}(M,\textit{efc})
\bar{\delta}_{i,\ell}(M,\textit{efc})
+
C_{\mathrm{impl}},
\label{eq:search_cost_coarse}
\end{equation}
where $c_d$ captures the cost of candidate inspection and distance
evaluation for a fixed $d$-dimensional workload, and $C_{\mathrm{impl}}$
absorbs backend-specific bookkeeping overheads. This accounting abstracts
away implementation constants. In particular, implementations that
maintain visited sets may avoid repeated distance computations for the
same node, so adjacency inspection and distance evaluation need not be in
one-to-one correspondence. The structural dependence is nevertheless
unchanged: $\textit{efc}$ primarily controls the expansion term, whereas
$M$ controls the degree cap that bounds the realized-degree term.

\textbf{Role of $\textit{efc}$.}
The construction parameter $\textit{efc}$ bounds the candidate set
maintained by SEARCH-LAYER and therefore controls the breadth of graph
exploration during insertion. Increasing $\textit{efc}$ allows the search
to retain more candidates and expand more nodes before termination. In
Equation~\eqref{eq:search_cost_coarse}, this primarily increases the
expansion term $X_{i,\ell}(M,\textit{efc})$, explaining the build-time
increase observed in Figure~\ref{fig:build_time_by_m_efc}(a). A broader
construction search can also expose SELECT-NEIGHBORS to more candidate
neighbors, which explains the accompanying increase in realized
connectivity.


\textbf{Saturation effects with increasing $M$.}
The parameter $M$ sets the per-layer degree cap up to a constant-factor
difference at layer $0$. To quantify realized connectivity, we define the
average out-degree across the HNSW layer hierarchy as
\begin{equation}
\mathrm{AOD}(M,\textit{efc})
=
\frac{
\sum_{\ell}\sum_{v\in V^{(\ell)}} \deg^{(\ell)}(v)
}{
\sum_{\ell}|V^{(\ell)}|
}.
\label{eq:aod}
\end{equation}
AOD is the realized counterpart of the degree cap: it measures how many
neighbors are retained after diversification and re-pruning. Although the
degree cap grows with $M$, realized AOD need not grow linearly. HNSW's
neighbor selection procedure favors diverse neighbors and may reject
redundant candidates even when the cap increases; when existing
neighborhoods exceed their degree limits, the same pruning logic is
applied again. As a result, increasing $M$ yields diminishing marginal
increases in realized connectivity. 

In Equation~\eqref{eq:search_cost_coarse}, this means that
$\bar{\delta}_{i,\ell}(M,\textit{efc})$ grows sublinearly with the cap,
which is consistent with the build-time saturation observed in
Figure~\ref{fig:build_time_by_m_efc}(b).

Overall, HNSW construction induces a separable resource dependence:
$\textit{efc}$ primarily governs the breadth of construction search, while
$M$ controls per-expansion connectivity through the realized graph degree.
Section~\ref{sec:constraints_estimator} instantiates this structure as a
compact build-time feasibility estimator.



\subsection{Index Size}
\label{sec:index_size_modeling}

\begin{figure}[t!]
  \centering
  \begin{subfigure}[b]{0.5\linewidth}
    \includegraphics[width=\linewidth]{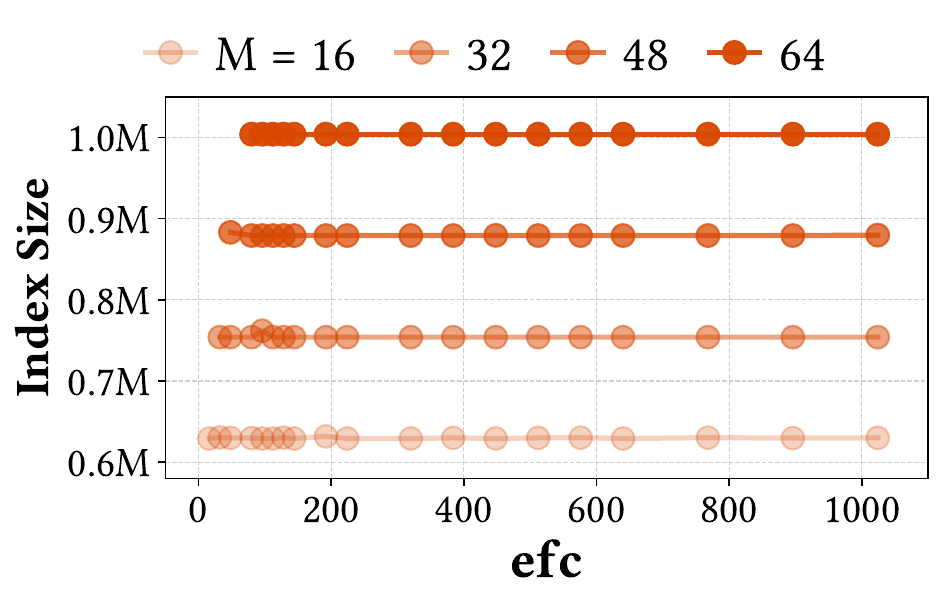}
    \vspace{-1.5em}
    \caption{}
    \label{fig:index_size_by_efc}
  \end{subfigure}\hfill
  \begin{subfigure}[b]{0.5\linewidth}
    \includegraphics[width=\linewidth]{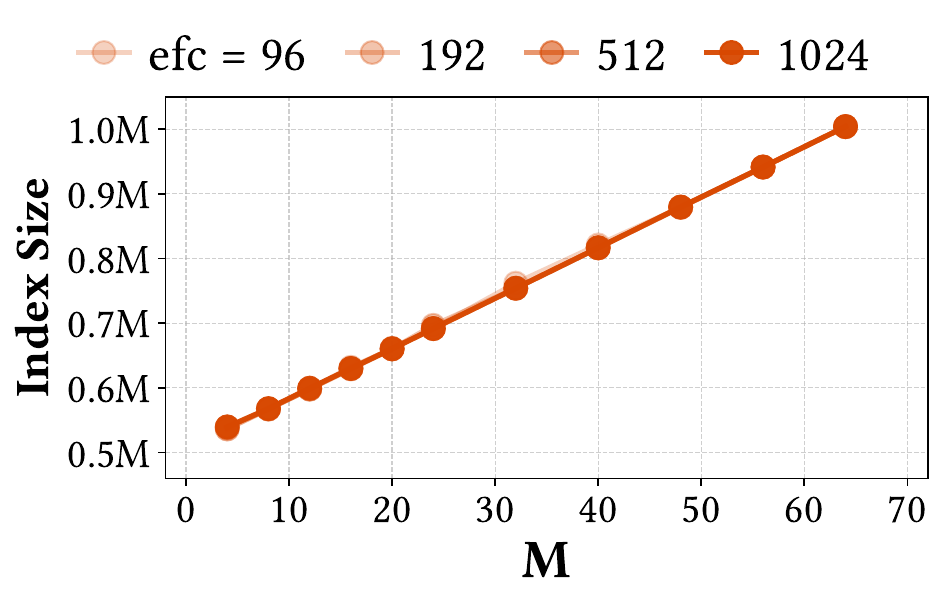}
    \vspace{-1.5em}
    \caption{}
    \label{fig:index_size_by_m}
  \end{subfigure}
  \vspace{-2em}
  \caption{Index size dependence on construction hyperparameters. Index size as a function of $(M,\textit{efc})$. \jae{(a) For a fixed $M$, index size is largely insensitive to $\textit{efc}$. (b) For a fixed $\textit{efc}$, index size increases approximately linearly with $M$, and curves for different $\textit{efc}$ largely overlap.}}
  \label{fig:index_size_by_m_efc}
\end{figure}

For fixed $N$ and $d$, vector storage and metadata are constant over
$(M,\mathit{efc})$; the hyperparameter-dependent term is adjacency storage.
Since HNSW stores $O(M)$ outgoing edges for an expected $O(N)$ node copies
across layers, the dominant index-size dependence is
\begin{equation}
S_{\mathrm{index}}(N,M)
=
C_{\mathrm{fixed}}(N,d)
+
\Theta(NM),
\label{eq:index_size_scaling}
\end{equation}
where $C_{\mathrm{fixed}}(N,d)$ includes vector storage and fixed metadata.
Backend-specific choices such as memory layout, alignment, and auxiliary
bookkeeping change the constants, but not the qualitative dependence on
$(M,\textit{efc})$.

The construction parameter $\textit{efc}$ does not appear in
Equation~\eqref{eq:index_size_scaling} because it changes the search
budget used during construction, not the degree cap that determines
adjacency capacity. A larger $\textit{efc}$ may affect which neighbors are
selected, and can therefore slightly change realized connectivity, but it
does not change the dominant storage budget induced by $M$. Thus, unlike
build time, index size is governed primarily by the degree-cap-induced
adjacency storage rather than by the amount of construction search. This
explains why index size remains largely invariant to $\textit{efc}$ in
Figure~\ref{fig:index_size_by_m_efc}(a), while increasing approximately
linearly with $M$ in Figure~\ref{fig:index_size_by_m_efc}(b).

\subsection{Feasibility-Aware Tuning Implications}
\label{sec:resource_implications}

The analysis above yields three structural properties of HNSW resource
usage:

\begin{itemize}[leftmargin=*, labelsep=0.5em, topsep=0pt, itemsep=0pt]
\item construction time increases with the search budget $\textit{efc}$.
\item construction time grows sublinearly with $M$ because realized
connectivity saturates under neighbor selection and pruning.
\item index size scales approximately linearly with $M$ and is largely
independent of $\textit{efc}$.
\end{itemize}

These separable dependencies allow infeasible regions of the
hyperparameter space to be identified without constructing every candidate
index. Build-time constraints can be evaluated using the dependence of
search breadth on $\textit{efc}$ and realized connectivity on $M$, while
index-size constraints can be evaluated primarily from $M$.
Section~\ref{sec:constraints_estimator} uses these structural properties
to construct compact feasibility estimators, and
Section~\ref{sec:solution} applies them to prune resource-violating
configurations before expensive index construction.

%% file: sections/solution.tex
\section{\reviewerB{Efficient Constraint-Aware Hyperparameter Optimization}}
\label{sec:solution}

This section presents \frameworkname{}, a constraint-aware framework for tuning HNSW hyperparameters. Given user constraints on search quality, search efficiency, resource usage, and tuning time, \frameworkname{} seeks the configuration $(M, efc, efs)$ that maximizes the target objective while satisfying these constraints.

\frameworkname{} combines three ideas: it decomposes HNSW tuning into structured subproblems, exploits the monotonicity and dominant unimodality analyzed in Section~\ref{sec:hyperparameters_performance}, and filters resource-infeasible configurations using the models in Section~\ref{sec:hyperparameters_resource}.

\reviewerC{
\frameworkname{} is API-level black-box but HNSW-structure-aware. It uses only standard build/query APIs and validation-workload measurements, including Recall, QPS, build time, and index size, without inspecting adjacency lists, layer populations, or degree distributions. However, it is not ANN-agnostic: its search decomposition relies on public HNSW mechanisms, namely bounded-degree construction controlled by M, greedy layer-wise \textsc{Search-Layer}, and diversification-based \textsc{Select-Neighbors}. We therefore scope its structural claims to standard CPU HNSW implementations that preserve these semantics, including Faiss~\cite{8733051}, Hnswlib~\cite{malkov2018efficientrobustapproximatenearest}, and Milvus~\cite{10.1145/3448016.3457550}. These assumptions may weaken for variants with different neighbor selection, disk- or GPU-dominated bottlenecks, or dynamic graph maintenance. They also do not automatically transfer to non-HNSW families such as IVF+PQ~\cite{5432202}, ScaNN~\cite{pmlr-v119-guo20h}, or NSG~\cite{fu12fast}.
}

\begin{figure}[t]
  \centering
  \includegraphics[width=1\linewidth]{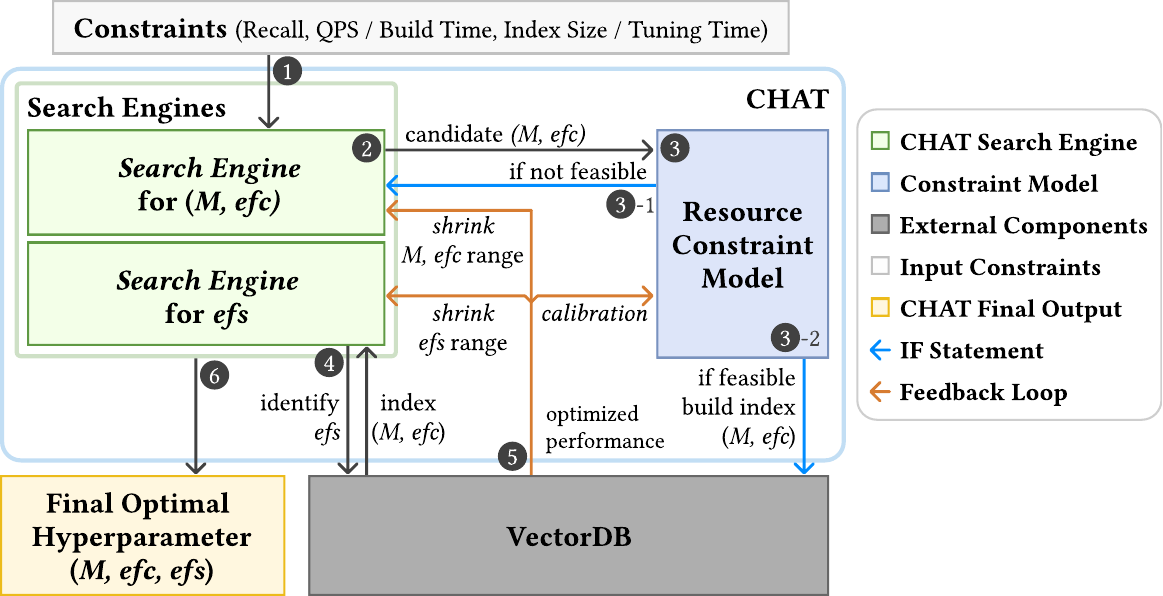}
  \caption{The overall architecture and iterative tuning workflow of \frameworkname{}.}
  \label{fig:overview_architecture}
\end{figure}

\subsection{System Overview}
\label{sec:overview}

\gun{
This subsection outlines the end-to-end tuning workflow of \frameworkname{}.
Figure~\ref{fig:overview_architecture} illustrates the system architecture and
the iterative interaction among its core components.
The workflow begins when \blackcircled{1} the user specifies a target dataset,
an HNSW-based vector database backend, and a set of constraints, including 
performance requirements (e.g., minimum Recall or QPS), resource limits on
build time and index size, and a total tuning-time budget.
These inputs jointly define the optimization objective and the feasible search
space.
\blackcircled{2} Guided by the user inputs, the \textit{Search Engine for
$(\textit{M}, \textit{efc})$} proposes a candidate construction configuration $(\textit{M}, \textit{efc})$.
This component orchestrates the outer loop of the tuning process, exploring
the structural parameter space.
\blackcircled{3} The proposed candidate is immediately evaluated by the 
\textit{Resource Constraint Models}, which predicts whether the configuration is likely 
to satisfy the specified resource constraints using learned cost models.
If the candidate is predicted to be infeasible, \blackcircled{3}-1 the search engine
receives feedback to prune the corresponding region of the search space without 
incurring actual build costs.
Otherwise, \blackcircled{3}-2 the system proceeds to construct an HNSW index on 
the target backend using the candidate $(\textit{M}, \textit{efc})$.
\blackcircled{4} Once the index is built, the \textit{Search Engine for
\efs{}} takes over to identify the optimal \efs{} value that maximizes the 
 objective performance while adhering to the performance constraint.}
\hoeun{\blackcircled{5} The measured performance and resource usage are used to calibrate the resource constraint models, thereby shrinking the search space in subsequent iterations.}
\gun{Finally, \blackcircled{6} once the tuning budget is exhausted or the search 
converges, \frameworkname{} reports the final optimal hyperparameter configuration 
$(M, \textit{efc}, \textit{efs})$.
Through this iterative process, \frameworkname{} efficiently converges to a 
near-optimal configuration.
}
\cgcg{
When the indexed corpus or query workload subsequently changes, CHAT
reuses the same history as a cross-run warm start, as detailed in
Section~\ref{sec:drift_aware_retuning}.
}

\subsection{Search for Optimal Construction Configuration}
\label{sec:M_efc_search}

\gun{
\frameworkname{} searches the construction parameters using a nested but simple strategy.
For each probed \M{}, it searches for a promising \efc{}, skips candidates that are model-infeasible under the resource budgets, builds only the remaining candidates, and then selects \efs{} by boundary search on the built index.
Thus, the original three-dimensional tuning problem is reduced to three one-dimensional decisions: localizing a good \M{} region, selecting the best \efc{} for that \M{}, and choosing the boundary \efs{} that satisfies the active performance constraint.
We write $g(M)$ for the best validated objective value returned after optimizing \efc{} and \efs{} under a fixed \M{}; formally, $g(M)=\max_{efc}\textit{perf}(M,efc,efs^*(M,efc))$ over configurations feasible under $\mathcal{C}_{\text{perf}}$ and $\mathcal{C}_{\text{res}}$.
}

\begin{algorithm}[t]
\caption{Search Engine for $(\M{},\efc{})$}
\label{alg:chat_overview_merged}
\footnotesize
\begin{algorithmic}[1]
\Input Dataset $\mathcal{D}$; HNSW backend $\mathcal{V}$;
performance constraint $\mathcal{C}_{\text{perf}}$;
resource constraints $\mathcal{C}_{\text{res}}$;
tuning budget $\mathcal{T}$
\Output Best validated configuration $bestCfg=(\M{},\efc{},\efs{})$

\State Initialize resource models $\hat{h}_{bt},\hat{h}_{is}$ and history $\mathcal{H}\leftarrow\emptyset$
\State Initialize domains $\mathcal{M}\leftarrow[4,64]$, $\mathcal{E}_c\leftarrow[8,1024]$
\State $bestCfg\leftarrow\bot$, $bestPerf\leftarrow-\infty$

\While{remaining tuning time $>0$ and $\mathcal{M}$ not converged}
    \State Choose ternary probes $\M{}_1,\M{}_2\in\mathcal{M}$
    \For{$\M{}\in\{\M{}_1,\M{}_2\}$}
        \State $(perf_M,cfg_M)\leftarrow$
        \Statex \qquad \textsc{Optimize-EFC}$(\M{},\mathcal{E}_c,\mathcal{D},\mathcal{V},
        \mathcal{C}_{\text{perf}},\mathcal{C}_{\text{res}},\mathcal{T},
        \hat{h}_{bt},\hat{h}_{is},\mathcal{H})$
        \If{$perf_M>bestPerf$}
            \State $bestPerf\leftarrow perf_M$; $bestCfg\leftarrow cfg_M$
        \EndIf
    \EndFor
    \State Update $\mathcal{M}$ toward the probe with higher validated performance
\EndWhile
\State \Return $bestCfg$
\end{algorithmic}
\end{algorithm}

\gun{
Algorithm~\ref{alg:chat_overview_merged} implements the outer search over \M{}.
Each probed \M{} is evaluated by \textsc{Optimize-EFC}, which returns the best measured feasible triplet found for that \M{}.
The outer loop then shrinks the \M{} interval toward the probe with higher validated performance, following the dominant unimodal trend of $g(M)$ analyzed in Section~\ref{sec:performance_across_M}.
The best measured feasible configuration is tracked throughout the search, so \frameworkname{} can return a validated configuration even if the tuning budget expires before the intervals fully converge.
}

\begin{algorithm}[t]
\caption{\textsc{Optimize-EFC}}
\label{alg:ternary_over_efc}
\footnotesize
\begin{algorithmic}[1]
\Input Fixed \M{}; \efc{} domain $\mathcal{E}_c$; dataset $\mathcal{D}$;
backend $\mathcal{V}$; performance constraint $\mathcal{C}_{\text{perf}}$;
resource constraints $\mathcal{C}_{\text{res}}$; tuning budget $\mathcal{T}$;
models $\hat{h}_{bt},\hat{h}_{is}$; history $\mathcal{H}$
\Output Best validated score $g(M)$ and configuration $cfg_M$

\State Shrink $\mathcal{E}_c$ using history $\mathcal{H}$
\State $g(M)\leftarrow-\infty$, $cfg_M\leftarrow\bot$

\While{remaining tuning time $>0$ and $\mathcal{E}_c$ not converged}
    \State Choose ternary probes $\efc{}_1,\efc{}_2\in\mathcal{E}_c$
    \For{$\efc{}\in\{\efc{}_1,\efc{}_2\}$}
        \If{$\hat{h}_{bt}(\M{},\efc{})$ or $\hat{h}_{is}(\M{})$ violates $\mathcal{C}_{\text{res}}$}
            \State Record model-infeasibility in $\mathcal{H}$; \textbf{continue}
        \EndIf

        \State $I\leftarrow$ \textsc{BuildIndex}$(\mathcal{V},\mathcal{D},\M{},\efc{})$
        \State Measure build time and index size; update $\hat{h}_{bt},\hat{h}_{is}$

        \If{measured resource usage violates $\mathcal{C}_{\text{res}}$}
            \State Record measured infeasibility in $\mathcal{H}$; \textbf{continue}
        \EndIf

        \State $(\efs{}^*,perf)\leftarrow\textsc{Optimize-EFS}(I,\mathcal{C}_{\text{perf}},\mathcal{H})$
        \If{$\efs{}^*=\bot$}
            \State Record performance infeasibility in $\mathcal{H}$; \textbf{continue}
        \EndIf

        \State Record $(\M{},\efc{},\efs{}^*,perf)$ in $\mathcal{H}$
        \If{$perf>g(M)$}
            \State $g(M)\leftarrow perf$; $cfg_M\leftarrow(\M{},\efc{},\efs{}^*)$
        \EndIf
    \EndFor
    \State Update $\mathcal{E}_c$ toward the probe with higher validated performance
\EndWhile
\State \Return $(g(M),cfg_M)$
\end{algorithmic}
\end{algorithm}

\gun{
Algorithm~\ref{alg:ternary_over_efc} performs the inner search over \efc{} for a fixed \M{}.
Before ternary search, \frameworkname{} uses prior optima in $\mathcal{H}$ to narrow the candidate range: larger \M{} values typically need smaller construction effort because additional degree capacity already improves graph connectivity.
This rule only reduces the search range; the final choice is still made using measured resource feasibility and validation performance.
}

\gun{
The resource models and validation measurements play different roles.
The models are used only to avoid unnecessary index builds under the conservative resource test.
Candidates that pass this filter are still built, measured for build time and index size, and then evaluated by \textsc{Optimize-EFS}.
Thus, early pruning reduces tuning cost, while the reported configuration is selected only from configurations explicitly validated against the user constraints.
}

\subsection{Search for Optimal \efs{}}
\label{sec:efs_search}

\gun{
Given a fixed constructed index $I(\M{},\efc{})$, \frameworkname{} only needs to tune the search-time parameter \efs{}.
Unlike \M{} and \efc{}, changing \efs{} does not alter the graph, so all probes reuse the same built index.
As analyzed in Section~\ref{sec:effect_of_efs}, \efs{} induces an ordered feasibility boundary: larger \efs{} broadens base-layer traversal, improving Recall but increasing query cost.
\frameworkname{} therefore finds the boundary \efs{} by binary search rather than by rebuilding or re-exploring construction parameters.
}

\begin{algorithm}[t]
\caption{\textsc{Optimize-EFS}}
\label{alg:search_engine_efs}
\footnotesize
\begin{algorithmic}[1]
\Input Pre-built index $I(\M{},\efc{})$; performance constraint $\mathcal{C}_{\text{perf}}$; history $\mathcal{H}$
\Output Boundary value $\efs{}^*$ and measured objective $perf$

\State $[efs_{\min},efs_{\max}] \leftarrow$ shrink interval using $\mathcal{H}$
\Statex \hfill // history-based interval

\If{$\mathcal{C}_{\text{perf}}$ is $Recall \ge Recall_{\min}$}
  \State $\efs{}^*, perf \leftarrow$
  \textsc{BSearchMinFeasibleEFS}$(I,[efs_{\min},efs_{\max}],Recall_{\min})$
\ElsIf{$\mathcal{C}_{\text{perf}}$ is $QPS \ge QPS_{\min}$}
  \State $\efs{}^*, perf \leftarrow$
  \textsc{BSearchMaxFeasibleEFS}$(I,[efs_{\min},efs_{\max}],QPS_{\min})$
\EndIf

\If{no feasible \efs{} exists}
  \State \Return $(\bot,-\infty)$
\EndIf
\State \Return $(\efs{}^*, perf)$
\end{algorithmic}
\end{algorithm}

\gun{
Algorithm~\ref{alg:search_engine_efs} handles the two main performance constraints symmetrically.
Under a Recall constraint, \frameworkname{} returns the smallest feasible \efs{}, which maximizes QPS while meeting the Recall target.
Under a QPS constraint, it returns the largest feasible \efs{}, which maximizes Recall while meeting the throughput target.
The resulting pair $(\efs{}^*,perf)$ is returned to \textsc{Optimize-EFC} and used to compare candidate construction settings.
}

\gun{
The initial search interval is tightened using previous boundary values stored in $\mathcal{H}$.
Intuitively, stronger construction settings often require no larger \efs{} to satisfy the same performance constraint, because better graph connectivity reduces the query-time exploration needed for Recall and can also tighten the QPS-feasible upper bound.
Formally, for a prior configuration $(M',efc')$ with boundary value $\efs{}^*_{M',efc'}$, if $M'\ge M$ and $efc'\ge efc$, then $\efs{}^*_{M,efc}\ge \efs{}^*_{M',efc'}$; reversing both inequalities gives the corresponding upper bound.
\frameworkname{} uses this history only to reduce the number of probes; the returned \efs{} is still selected by measured feasibility on the current index.
Thus, history accelerates binary search but does not replace validation.
}

\reviewerA{
The same boundary-search view supports query-subset constraints.
For example, a deployment may require both average Recall and hard-query Recall:
\(\mathrm{Recall}(Q_{\mathrm{all}})\ge \rho_{\mathrm{all}}\) and
\(\mathrm{Recall}(Q_{\mathrm{hard}})\ge \rho_{\mathrm{hard}}\).
Because Recall on any fixed query subset is monotone in \efs{}, the conjunction remains a monotone predicate and can be handled by the same binary search.
We identify \(Q_{\mathrm{hard}}\) using Steiner-hardness~\cite{wang2024steiner}.
}

\reviewerC{
Tail-latency constraints can be handled analogously.
For a fixed index and validation workload, \frameworkname{} measures the requested tail metric, such as P95 or P99 latency, at each probed \efs{}.
Because increasing \efs{} only enlarges the \textsc{Search-Layer} candidate budget, the permitted traversal work is monotone non-decreasing in \efs{}; consequently, a fixed-workload latency quantile defines an upper-bound feasibility predicate up to measurement noise.
A service-level constraint such as \(\mathrm{P}_{\alpha}(\mathrm{Latency}) \le \tau_{\alpha}\) is therefore handled using the same max-feasible binary search as \(QPS \ge QPS_{\min}\), returning the largest \efs{} whose measured tail latency satisfies the constraint.
If Recall lower bounds and tail-latency upper bounds are specified together, 
the feasible \efs{} values form an interval, provided that the minimum Recall-feasible \efs{} does not exceed the maximum latency-feasible efs. 
The selected point in this interval depends on the active objective.
}


\subsection{\hoeun{Verifying Resource Constraints Satisfaction}}
\label{sec:constraints_estimator}

\hoeun{Efficient hyperparameter tuning requires an early mechanism to rule out configurations that violate user-specified resource budgets. Performance feasibility (QPS/Recall) is evaluated during search on a fixed index (Section~\ref{sec:hyperparameters_performance} and Algorithm~\ref{alg:search_engine_efs}). This subsection addresses \emph{resource} feasibility. Build time and index size budgets cannot be verified for every candidate $(M,\textit{efc})$ by fully constructing an index, as doing so would dominate tuning cost. We therefore employ a lightweight \emph{resource-feasibility filter} that instantiates the algorithm-driven dependencies derived in Section~\ref{sec:hyperparameters_resource} into closed-form surrogates and continuously calibrates their remaining constants using measurements obtained during the search process.}

\subsubsection{\hoeun{Closed-form Instantiation}}
\label{sec:constraints_instantiation}

\hoeun{Section~\ref{sec:hyperparameters_resource} establishes two structural properties. First, build time is governed by a search-breadth component controlled by $\textit{efc}$ and a per-expansion scan component governed by realized out-degree, whose marginal growth in $M$ saturates due to neighbor diversification and pruning. Second, index size is dominated by the cap-induced adjacency capacity and therefore scales linearly with $M$ while remaining invariant to $\textit{efc}$. We instantiate these dependencies using compact parametric surrogates that preserve monotonicity and saturation while leaving scale and offset to be determined by calibration.

We use two resource surrogates: $h_{\mathrm{bt}}(M,\textit{efc})$ for build time and $h_{\mathrm{is}}(M)$ for index size. For build time, we adopt a log-like surrogate in $M$ to represent saturating realized connectivity, while retaining an explicit dependence on $\textit{efc}$ to capture increased exploration: \begin{equation}
h_{\mathrm{bt}}(M,\textit{efc})
=
\alpha_0 + \alpha_1\,\textit{efc}
+ (\beta_0+\beta_1\,\textit{efc})\cdot \log_2(1+M)
+ \gamma\,\textit{efc}\cdot \log_2^2(1+M).
\label{eq:hbt}
\end{equation}
The linear term in $\textit{efc}$ captures the monotone increase in construction search budget (Figure~\ref{fig:build_time_by_m_efc}(a)). The $\log_2(1+M)$ terms represent a representative monotone concave dependence on $M$, consistent with the saturation of realized out-degree observed empirically (Figure~\ref{fig:build_time_by_m_efc}(b)). Logarithmic growth is a convenient member of a broader class of concave monotone functions consistent with the algorithmic upper bounds derived in Section~\ref{sec:build_time_modeling}; it preserves sublinear growth while remaining analytically simple and stable under incremental calibration from a limited number of observations. The interaction term $\textit{efc}\log_2^2(1+M)$ captures that higher connectivity increases the number of candidate evaluations performed per unit of exploration during construction.

For index size, Section~\ref{sec:index_size_modeling} implies a linear dependence on $M$ and no dependence on $\textit{efc}$: \begin{equation}
h_{\mathrm{is}}(M)=\eta_0+\eta_1\,M.
\label{eq:his}
\end{equation}
Although $\textit{efc}$ may influence realized connectivity (Figure~\ref{fig:build_time_by_m_efc}(a)), it does not affect the cap-induced adjacency budget that dominates index size (Figure~\ref{fig:index_size_by_m_efc}); accordingly, $\textit{efc}$ is excluded from $h_{\mathrm{is}}$.}

\subsubsection{Online Calibration from Search-time Measurements}
\label{sec:anchor_calibration}

\hoeun{The functional forms in \eqref{eq:hbt} and \eqref{eq:his} are fixed by the algorithmic structure; only their scale and offset parameters are backend-dependent. These parameters are calibrated using all configurations that are actually constructed during the search process. Whenever the search engine builds an index to evaluate performance feasibility, the resulting build time and index size are recorded and incorporated into calibration.

Let $\mathcal{C}$ denote the set of configurations $(M_j,\textit{efc}_j)$ that have been constructed so far during search, with measured build times $T_j$ and index sizes $S_j$. Define the feature maps \begin{equation}
\phi_{\mathrm{bt}}(M,\textit{efc})
=
\big[1,\ \textit{efc},\ \log_2(1+M),\ \textit{efc}\log_2(1+M),\ \textit{efc}\log_2^2(1+M)\big],
\end{equation}\begin{equation}
\phi_{\mathrm{is}}(M)=[1,\ M].
\end{equation}
The surrogate parameters are obtained by solving linear least-squares problems over $\mathcal{C}$: \begin{equation}
h_{\mathrm{bt}}=\langle \theta_{\mathrm{bt}},\phi_{\mathrm{bt}}\rangle,
\qquad
h_{\mathrm{is}}=\langle \theta_{\mathrm{is}},\phi_{\mathrm{is}}\rangle,
\end{equation} where $\theta_{\mathrm{bt}}=(\alpha_0,\alpha_1,\beta_0,\beta_1,\gamma)$ and $\theta_{\mathrm{is}}=(\eta_0,\eta_1)$. Because calibration is linear in the number of accumulated measurements, its overhead is negligible relative to index construction.

As search progresses, the calibration set $\mathcal{C}$ naturally expands to cover a wider range of $(M,\textit{efc})$ values, including regions where saturation in $M$ becomes apparent. This incremental calibration improves alignment between the surrogates and backend-specific constants without altering the surrogate structure.}

\subsubsection{\reviewerA{Conservative Feasibility Test}}
\label{sec:conservative_test}


\reviewerA{Let $T_{\max}$ and $S_{\max}$ denote user-specified budgets for build time
and index size. To avoid accepting configurations whose resource usage is
underestimated by the calibrated surrogates, CHAT applies one-sided safety
margins to the build-time and index-size estimates. Each margin is computed
from the positive relative underestimation residuals observed on constructed
configurations in $\mathcal{C}$ and is recomputed whenever $\mathcal{C}$
grows. A candidate is pruned before construction only when the margin-adjusted
build time or index size estimate exceeds the corresponding budget. This is a conservative false-feasible guard, not a deterministic no-pruning guarantee; larger margins reduce false-feasible risk at the cost of pruning aggressiveness.}

\subsection{\reviewerA{Drift-Aware Retuning}}
\label{sec:drift_aware_retuning}

\reviewerA{ \frameworkname{} supports retuning when the indexed corpus or query workload changes after an initial tuning run. Rather than restarting from scratch, it reuses the previous tuning history \(\mathcal{H}\) as a cross-run warm start. The system refreshes the validation workload using recent queries and the current corpus, monitors Recall, QPS, and the \efs{} feasibility boundary, and first reruns only \textsc{Optimize-Efs} at the previous \((\M{},\efc{})\). Construction retuning is invoked only when the deployed configuration becomes infeasible or when changing \efs{} alone is insufficient. 

When construction retuning is needed, \frameworkname{} applies \emph{Constraint-Directed Hard Pruning with Sentinel Safety} (CDHP). CDHP uses the feasibility-boundary shift in Section~\ref{sec:drift-implications} to predict the movement of the best construction effort \(efc^*_M\). Under a Recall constraint, a tighter boundary, i.e., a larger minimum \efs{} to reach the Recall target, indicates that graph quality has become more important and suggests moving toward larger \efc{}. Under a QPS constraint, a tighter boundary, i.e., a smaller maximum feasible \efs{}, makes high-\efc{} configurations cost-dominated earlier and suggests moving toward smaller \efc{}. Relaxed boundaries induce the opposite directions. 

Guided by this signal, CDHP keeps only the corresponding one-sided construction region: \(\efc{}\ge efc^{*,old}_M\) for a rightward move and \(\efc{}\le efc^{*,old}_M\) for a leftward move. For unseen \M{} values, the cutoff is inferred from the historical envelope of \(efc^*_M\). \frameworkname{} then runs the standard \textsc{Optimize-Efc} and \textsc{Optimize-Efs} procedures within the retained domain, with the same resource filtering as in the initial run. Since the boundary-shift signal is structural rather than universal, CDHP places sentinel probes on the pruned side and expands the cached search domain if a sentinel is competitive within measurement tolerance; in the worst case, it completes the remaining search domain using cached measurements rather than restarting from scratch. For deletions, \frameworkname{} first cleanly rebuilds the updated corpus using the previous \((\M{},\efc{})\) to distinguish in-place deletion artifacts from true hyperparameter drift.
}

%% file: sections/experiment.tex
\section{Experiments}
\label{sec:experiment}

\subsection{Experimental Setup}
\label{sec:setup}

\gun{
We evaluated \frameworkname{} across multiple vector database backends, datasets, baselines, and metrics.
We used Hnswlib~\cite{malkov2018efficientrobustapproximatenearest} and Faiss~\cite{8733051} as reference HNSW implementations, as they expose the core behavior of HNSW with minimal system-level optimizations, allowing us to isolate the intrinsic performance characteristics analyzed in Section~\ref{sec:hyperparameters_performance}.
To assess real-world applicability, we additionally evaluated \frameworkname{} on Milvus~\cite{10.1145/3448016.3457550}, a production-grade vector database that uses HNSW internally.
Across all platforms, we varied three hyperparameters: the maximum node degree $M \in [4,64]$, the construction candidate size $\efc{} \in [8,1024]$, and the search candidate size $\efs{} \in [10,1024]$.}

We benchmarked \frameworkname{} on five representative datasets. Four
of them---\texttt{nytimes}, \texttt{glove}, \texttt{sift}, and
\texttt{deep1M}---were standard public datasets provided by
ANN-Benchmarks~\cite{AUMULLER2020101374}, widely adopted for evaluating
approximate nearest neighbor search (ANNS) algorithms. 
\reviewerA{In addition, we constructed the \texttt{youtube} dataset~\cite{youtubedataset}, 
a realistic video retrieval workload built by encoding extracted I-frames
from public high-view-count YouTube videos with OpenCLIP H-14~\cite{ilharco_gabriel_2021_5143773}, providing
a diverse and complementary evaluation point.
The characteristics of all datasets are summarized in Appendix~C, which further details the \texttt{youtube} construction pipeline, including the data source, frame extraction, and embedding model
}.


\begin{figure*}[t!]
  \centering
  \includegraphics[width=1\linewidth]{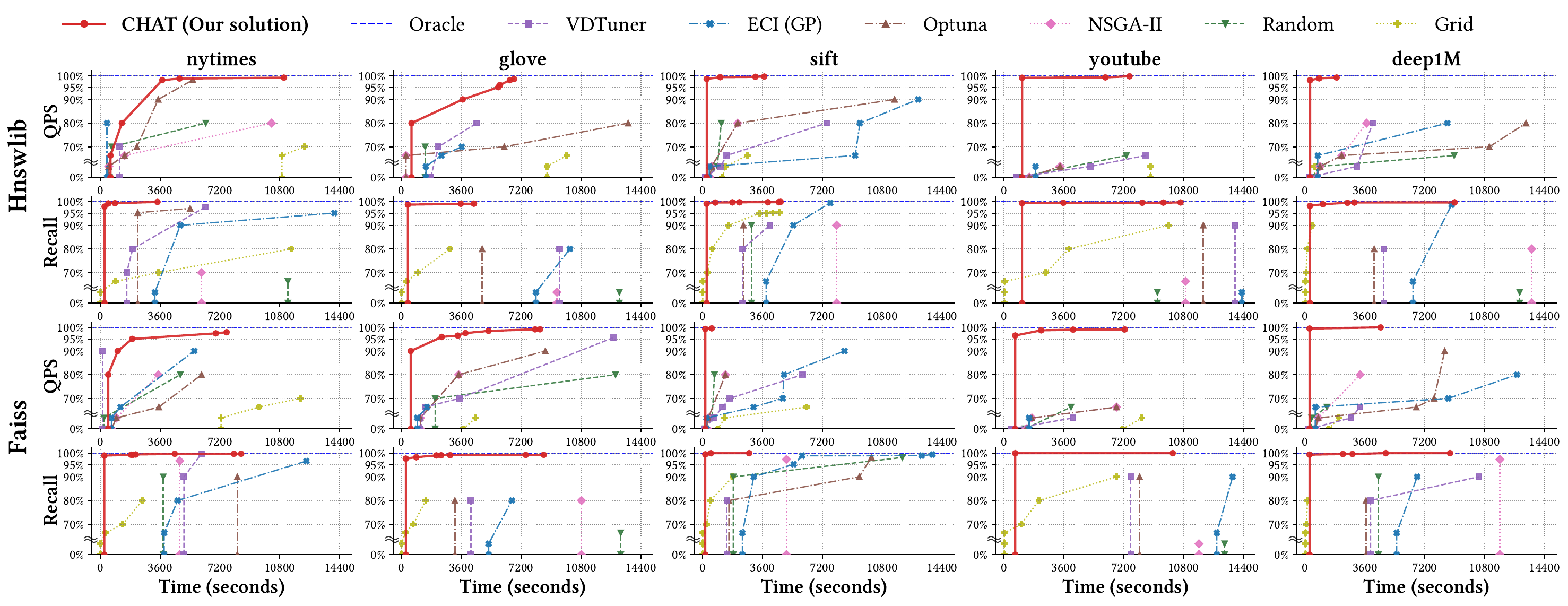}
    \caption{\gun{Tuning performance over time, showing the best-so-far objective value normalized to 
    the Oracle (exhaustive search). Each column corresponds to a dataset, and rows represent QPS (under Recall 
    constraint) and Recall (under QPS constraint) for Hnswlib and Faiss backends, respectively.
    \frameworkname{} consistently achieves faster convergence and higher final performance 
    than baseline methods across all datasets and constraint settings.}}
  \label{fig:main_figure}
  \vspace{-1em}
\end{figure*}

\reviewerA{
As comparison baselines, we considered a diverse set of black-box
hyperparameter tuning methods, covering model-free search, sequential
model-based optimization, evolutionary multi-objective search, constrained
Bayesian optimization, and vector-database tuning.
We included Random Search~\cite{10.5555/2188385.2188395} and
Grid Search~\cite{liashchynskyi2019gridsearchrandomsearch}, which
respectively sample configurations uniformly and enumerate a predefined
grid; both enforce constraints by discarding infeasible configurations.
We also evaluated Optuna~\cite{10.1145/3292500.3330701}, a TPE-based
sequential optimizer that proposes trials from past observations. In our
implementation, Optuna handles constraints through a large-penalty scalar
objective, assigning infeasible trials a penalty proportional to the degree
of violation.
NSGA-II~\cite{996017} is a population-based evolutionary optimizer that
uses Pareto dominance to balance objective value and constraint violation;
we implement it as a bi-objective search that minimizes violation while
maximizing the target metric.
We additionally included Expected Constrained Improvement
(ECI)~\cite{gardner2014bayesian}, a Gaussian-process constrained Bayesian
optimization method that models both objective and feasibility and selects
configurations by feasibility-weighted expected improvement.
We also compared against
VDTuner~\cite{yang2024vdtunerautomatedperformancetuning}, a
Bayesian-optimization framework designed for vector database tuning,
restricted to the same HNSW parameter space for fairness.
All baselines are treated as black-box tuners that do not inspect internal
HNSW graph statistics. 
A more detailed introduction to each baseline, together with its hyperparameter configurations and experimental settings, is provided in Appendix~D.
}

\gun{For evaluation, we measured Recall@10 and Queries Per Second (QPS) to assess search accuracy and throughput,
respectively, and recorded build time and index size to capture index construction cost and memory footprint.
All experiments were conducted on a single Ubuntu server with an AMD Ryzen Threadripper 3970X CPU
(32 cores, 64 threads, 3.7\,GHz), 188\,GB RAM, and 14\,GB swap, using CPU-only execution to reflect
typical vector database deployments.
To reduce runtime noise, each index configuration was built once and evaluated with the same query workload
over 10 runs, reporting the average Recall and QPS.}


\subsection{Constraint-Aware Tuning Results}

We evaluated the effectiveness of \frameworkname{} in identifying optimal
HNSW hyperparameter configurations under user-specified performance
constraints, using a common 4-hour wall-clock cap for all methods.
\reviewerA{
We considered two performance constraints: (1) $\mathrm{Recall}\ge0.95$, 
where QPS is maximized, and (2) $\mathrm{QPS}\ge\tau$, where Recall is 
maximized. \frameworkname{} always takes $\tau$ as a raw throughput 
target, e.g., $\mathrm{QPS}\ge5000$; the 75th percentile is used only 
in our evaluation to choose comparably difficult targets across 
dataset-backend pairs. In practice, $\tau$ can come from a service 
SLO or a small pilot probe---a quick QPS measurement at minimal 
$M$, $efc$, and $efs$ values that estimates the high-throughput 
scale of the workload. 
Appendix~E reports the raw QPS values used.
}

To assess the effectiveness of our search strategy, we first compared \frameworkname{}
against an Oracle (Exhaustive Search) that evaluates all valid hyperparameter configurations
within the predefined search ranges under the same performance constraint, providing an upper bound on achievable performance.
Despite this strong reference, \frameworkname{} achieved 98.03\%--99.97\% of the Oracle's QPS (under the Recall constraint) and 99.11\%--99.99\% of the Oracle's Recall (under the QPS constraint).
These results show that \frameworkname{} identifies near-optimal configurations with only a fraction of the exhaustive search cost, quantifying a small optimality gap under strict tuning budgets.

Figure~\ref{fig:main_figure} summarizes the tuning performance of all methods by plotting the best objective values
attained over time, normalized to the Oracle, across all datasets and constraint settings.
In all cases, \frameworkname{} (shown in red) either converged to the best-performing configuration
substantially faster than baseline methods or maintained superior objective performance throughout
the tuning process.
Under the Recall constraint (i.e., when maximizing QPS), \frameworkname{} outperformed all baseline
methods by margins ranging from 1\% to 45\% (1\% on the \texttt{nytimes} dataset with Hnswlib, and 45\%
on the \texttt{youtube} dataset with Hnswlib) in QPS. Under the QPS constraint (i.e., when maximizing
Recall), \frameworkname{} again demonstrated strong performance, achieving Recall improvements
of 0.6\% to 11\% (0.6\% on the \texttt{nytimes} dataset with Faiss, and 11\% on the \texttt{glove} dataset with
Faiss) over the best baselines.

\frameworkname{} converges quickly enough for practical deployment.
Using the time to reach 95\% of the Oracle objective as a convergence metric, \frameworkname{} is the only method to reach this threshold within the 4-hour cap in 13 out of 20 settings; in the remaining cases, it reaches the target $1.5\times$--$44\times$ faster than the best competing baseline.
\reviewerA{
The 4-hour cap is an evaluation budget rather than the runtime of 
\frameworkname{}: across the 10 backend--dataset combinations 
(two HNSW backends, Faiss and hnswlib, evaluated on five datasets) 
under the Recall=0.95 constraint, \frameworkname{} reaches 99\% of 
the Oracle objective in 48 minutes on average (range 13--123 minutes), 
consuming only 20\% of the cap.
Under a conservative no-serving-during-tuning break-even model, the mean 
break-even time is 3 hours and every backend--dataset combination recovers 
the tuning cost within 7 hours; over one week, the tuned configuration serves 
0.25--2 billion additional queries per combination.
}

\reviewerA{
Beyond average Recall and QPS, we also evaluate query-side SLOs
supported by \textsc{Optimize-EFS} in Section~\ref{sec:efs_search}.
For hard-query Recall, we rank validation queries by Steiner-hardness
and define Hard-20 as the hardest 20\% of queries. Average-only tuning
satisfies the all-query Recall target but leaves Hard-20 Recall below
0.95 in nine of ten Faiss/Hnswlib settings. When Hard-20 Recall is added
as an explicit constraint, \frameworkname{} satisfies both the all-query
and Hard-20 Recall targets in all ten settings while achieving
98.21\%--99.69\% of the corresponding Oracle objective.}
\reviewerC{
For tail latency, we evaluate raw P95 and P99 latency upper-bound
constraints. Across the same ten Faiss/Hnswlib settings, the structural
properties exploited by \frameworkname{} continue to hold empirically
under these query-side SLOs: the \efs{} feasibility predicate remains
monotone, and the fixed-\M{} objective over \efc{} and the outer objective
over \M{} continue to show dominant unimodal trends. Under the P95/P99
latency constraints, \frameworkname{} satisfies the requested tail-latency
bounds while achieving 99.07\%--99.99\% of the corresponding Oracle
objective.
}

\textbf{Why \frameworkname{} Outperforms.}
The superior performance and rapid convergence of \frameworkname{} stem from its
structure-exploiting design.
Unlike prior approaches (e.g., VDTuner, Optuna, ECI,
NSGA-II, and Random/Grid Search), which treat HNSW tuning as a black-box optimization problem,
\frameworkname{} explicitly leverages the monotonic and unimodal performance structure of HNSW.
By decomposing the search space into sequential one-dimensional subproblems and applying
deterministic binary and ternary search, \frameworkname{} progressively eliminates suboptimal
regions and can identify when further exploration is unlikely to yield meaningful improvement,
allowing it to terminate tuning early while still achieving near-optimal performance.
Lightweight resource constraint models (Section~\ref{sec:constraints_estimator}) further enable
pruning of infeasible regions without the overhead of full index construction.

In contrast, baseline methods suffer from a fundamental mismatch between their optimization
mechanisms and the structural properties of HNSW.
Structure-agnostic approaches such as Grid and Random Search incur excessive or unguided
exploration costs and provide no guarantees of convergence, making them unreliable under
tight tuning budgets when high-performing feasible configurations are sparse.
Bayesian optimization–based methods, including Optuna, VDTuner, and ECI, 
rely on probabilistic modeling or surrogate-based feasibility estimation
to localize the constraint boundary.
However, in the discrete and sharply structured HNSW configuration space, learning this
boundary stochastically requires expending a substantial portion of the tuning budget on
exploration.
Similarly, NSGA-II preserves population diversity to approximate a Pareto frontier across
objectives, which inherently dilutes selection pressure toward the single constrained
optimum required in our setting.
As a result, although these methods are powerful general-purpose optimizers, they lack the
domain-specific structural awareness needed to efficiently navigate feasibility boundaries
in constraint-critical HNSW hyperparameter tuning.

\subsection{\hoeun{Tuning under Resource Constraints}}
\label{sec:estimator_accuracy}

\begin{table}[t]
  \centering
  \footnotesize
  \begin{adjustbox}{max width=\linewidth}
  \begin{tabular}{@{} lcccc @{}}
    \toprule
    \textbf{Metric} & \multicolumn{2}{c}{\textbf{Index Size Constraint}} & \multicolumn{2}{c}{\textbf{Build Time Constraint}} \\
    \cmidrule(lr){2-3} \cmidrule(lr){4-5}
    \textbf{Engine} & \textbf{Faiss} & \textbf{Hnswlib} & \textbf{Faiss} & \textbf{Hnswlib} \\
    \midrule
    \textbf{(a) Search Process} \\
    Index Builds / Visited Configs & 29.8 / 38.3 & 25.5 / 33.1 & 34.8 / 38.3 & 29.4 / 33.1 \\
    Pruned Before Build & 22.44\% & 25.95\% & 4.20\% & 11.60\% \\ 
    Constraint-Violating Builds & 4.0 / 10.5 & 4.0 / 10.4 & 2.9 / 9.7 & 2.4 / 9.3 \\
    Feasibility Coverage & 0.9800 & 0.9817 & 0.9956 & 0.9873 \\
    
    \midrule
    \textbf{(b) Search Results} \\
    Total Search Time (s) & 4791 / 6531 & 5042 / 7557 & 4787 / 6531 & 5327 / 7557 \\
    Average Performance Gain & 1.36$\times$ & 1.50$\times$ & 1.36$\times$ & 1.42$\times$ \\
    Best Performance Gain & 6.26$\times$ & 5.21$\times$ & 2.57$\times$ & 2.32$\times$ \\
    Performance Preservation & 100.0\% & 100.0\% & 100.0\% & 100.0\% \\
    \bottomrule
  \end{tabular}
  \end{adjustbox}
  \caption{Impact of analytic resource-constraint models on the tuning process. All values are reported as \emph{with constraint model / without constraint model}. Resource-feasibility filtering substantially reduces unnecessary index constructions and total tuning time while fully preserving the final selected configuration.}
  \vspace{-2em}
  \label{tab:esimator_comparison}
\end{table}

This section evaluates the effect of analytic resource constraint models on the tuning process. In contrast to Section~6.2, which compares tuning strategies under performance constraints (recall and QPS), this section isolates the contribution of the resource-feasibility filter itself. The resource constraint model is treated as a black-box module that can be attached to any search strategy; evaluating it separately avoids biasing cross-baseline comparisons and allows its standalone impact to be assessed under identical tuning budgets. Table~\ref{tab:esimator_comparison} summarizes the results under two resource constraints (index size and build time) and two HNSW backends (Faiss and Hnswlib). All values are reported as \emph{with constraint model / without constraint model} and are averaged over the five datasets described in Section~\ref{sec:setup}.

\textbf{Search Process Efficiency.}
Table~\ref{tab:esimator_comparison}(a) summarizes the tuning process. \emph{Index Builds / Visited Configs} shows that resource-feasibility filtering substantially reduces the number of index constructions required to explore the same configuration space by pruning infeasible candidates prior to construction. This behavior is reflected in the \emph{Pruned Before Build} rate (4.2\%–26.0\%) and the consistently lower number of \emph{Constraint-Violating Builds} when constraint models are applied. The safety of early pruning is quantified by \emph{Feasibility Coverage}, computed over all configurations visited during the search. Coverage exceeds 0.98 across all settings, indicating that the constraint models remain conservative while effectively filtering infeasible regions.

\textbf{Impact on Search Outcomes.}
As shown in Table~\ref{tab:esimator_comparison}(b), \emph{Total Search Time}, which is dominated by index construction, is consistently lower when resource-feasibility filtering is enabled. This reduction is consistent with the decrease in index builds observed in Table~\ref{tab:esimator_comparison}(a), indicating that pruning infeasible configurations prior to construction is the primary source of the time savings.

Despite the reduced exploration cost, the quality of the final tuning result is fully preserved. As shown in Table~\ref{tab:esimator_comparison}(b), \emph{Performance Preservation} is 100\% across all settings, indicating that the configuration selected with resource constraints achieves identical performance to that selected without constraints. Across all five datasets, \emph{Average Performance Gain} remains stable in the range of $1.36\times$–$1.50\times$, showing that feasibility filtering does not restrict access to competitive configurations. 
\reviewerA{
Figure~\ref{fig:margin_sensitivity} isolates the safety-margin mechanism by
comparing no-margin, fixed-margin, and residual-based variants. No-margin
pruning is more exposed to false-feasible decisions, while fixed margins
require a hand-tuned conservatism level. Residual-based margins adapt to
observed underestimation residuals, reducing false-feasible risk while
maintaining high pruning rates. Together with the performance preservation and
best-case gains in Table~\ref{tab:esimator_comparison}, this shows that
resource-feasibility filtering lowers tuning cost without blocking
high-performing configurations.
}


\begin{figure}[t]
  \centering
  \includegraphics[width=0.95\linewidth]{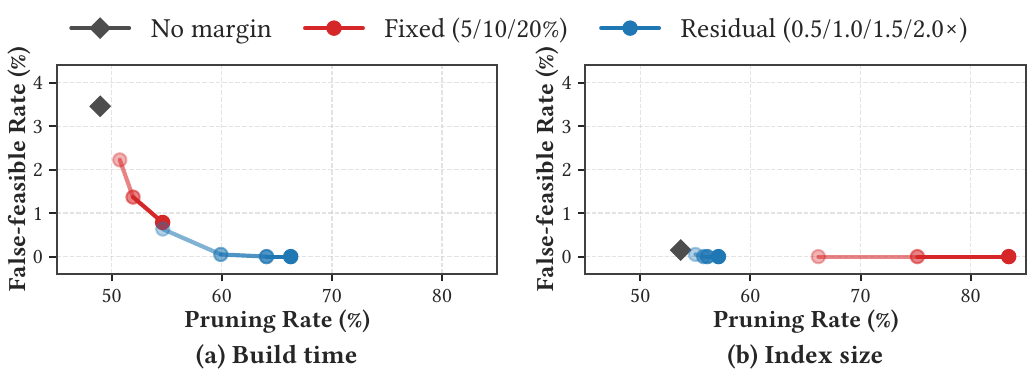}
\vspace{-1em}
\caption{\reviewerA{
Safety-margin sensitivity under resource constraints. The figure compares
no-margin, fixed-margin, and residual-based variants in terms of pruning rate
and false-feasible rate for build-time and index-size constraints. Marker opacity increases with the numeric margin value,
so darker markers indicate larger margins.
}}
  \label{fig:margin_sensitivity}
\end{figure}

\subsection{\reviewerA{Robustness of Structural Trends and Measurements}}
\label{sec:structural_robustness}

\frameworkname{} uses ternary-style localization over measured objective
values, but it does not assume convexity or universal unimodality.
Section~\ref{sec:hyperparameters_performance} already verifies that
changing query parallelism shifts absolute QPS but preserves the \efs{}
feasibility boundary and the dominant trends over \efc{} and \M{}.
Here we audit two further risks: whether measurement noise perturbs the
pairwise comparisons used during search, and whether the dominant
structural trends---formalized as sufficient conditions in
Lemmas~\ref{lem:efc_unimodal} and~\ref{lem:m_unimodal}---persist across
different data geometries.

We first evaluate robustness to measurement noise using the 10-run
average as a reference. We recompute configuration scores using 1, 3,
and 5 runs, and measure QPS \(P_{95}\) relative error, pairwise
comparison flip rate, same-or-adjacent optimum stability, and retained
objective value. Pairwise flips exclude candidate pairs whose reference
objective gap is below 0.5\%, so near-ties are not counted as unstable
decisions. Even with a single run, the QPS \(P_{95}\) error is 1.86\%,
the pairwise flip rate is only 1.1\%, the selected optimum remains the
same or adjacent grid point in 98.8\% of cases, and the selected
configuration retains 99.21\% of the 10-run reference objective. With
five runs, these improve to 0.62\%, 0.4\%, 99.7\%, and 99.79\%,
respectively. Thus, runtime noise does not materially perturb the
pairwise comparisons used by \frameworkname{}.

\reviewerC{
We further stress-test the structural trends on six synthetic workloads
that vary cluster separation/overlap, cluster-size imbalance, anisotropy,
dimensionality, and cluster structure. For the well-separated, overlapping,
imbalanced, and anisotropic workloads, the trends exploited by
\frameworkname{} continue to hold empirically under both Recall and QPS
constraints: the \efs{} feasibility predicate remains monotone,
fixed-\M{} slices over \efc{} remain dominantly unimodal, and the outer
envelope \(g(M)\) remains dominantly unimodal. In these cases, HNSW
construction exposes useful routing alternatives early, after which
additional construction effort becomes increasingly redundant.

The high-dimensional and uniform workloads expose the main failure mode.
Distance concentration or the absence of cluster structure can delay
useful routing edges until larger construction budgets, producing a
delayed graph-quality gain and a second objective peak. This effect is
most visible under a Recall constraint, where the delayed gain can
reduce the minimum feasible \efs{} and recover QPS; under a QPS
constraint, the same branch must reduce \efs{} to preserve throughput,
so the gain is dampened.
This behavior is consistent with the sufficient conditions in
Lemmas~\ref{lem:efc_unimodal} and~\ref{lem:m_unimodal};
it illustrates cases where the diminishing-utility
condition required by these lemmas fails. \frameworkname{} is therefore
not a global optimizer for arbitrary multi-modal objectives, but a
sample-efficient search procedure for dominant discrete trends induced
by standard HNSW mechanisms. In the high-dimensional counterexample,
\frameworkname{} selects \((M,efc,efs)=(16,400,160)\), while the full-grid
Oracle selects \((32,660,320)\); both satisfy the Recall constraint.
\frameworkname{} achieves QPS 4{,}970 versus the Oracle's 5{,}280,
retaining 94.1\% of the Oracle objective. Thus, when the
sufficient conditions fail, ternary-style localization may converge to
a local peak, but best-so-far tracking over validated configurations
substantially limits the loss.
}

\begin{figure}[t!]
  \centering
  \begin{subfigure}[b]{0.48\linewidth}
    \includegraphics[width=\linewidth]{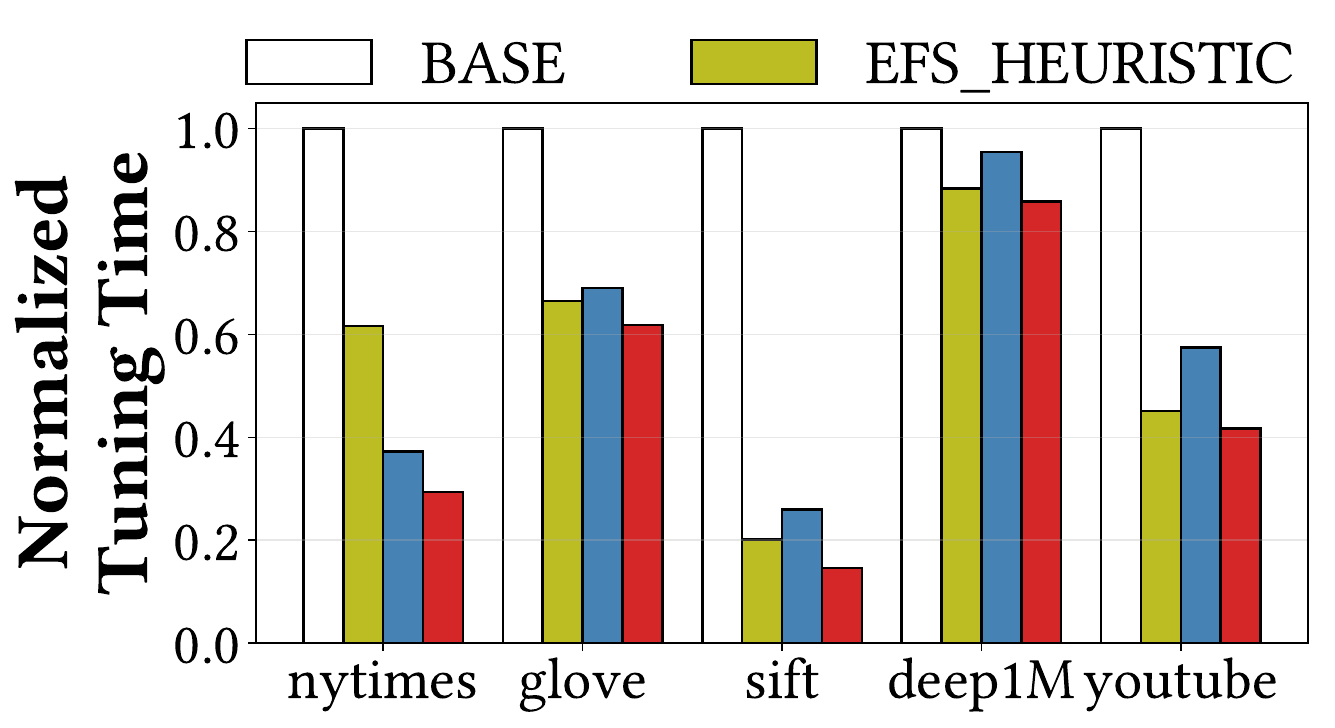}
    \vspace{-1.5em}
    \caption{}
    \label{fig:ablation_sub_a}
  \end{subfigure}\hfill
  \begin{subfigure}[b]{0.48\linewidth}
    \includegraphics[width=\linewidth]{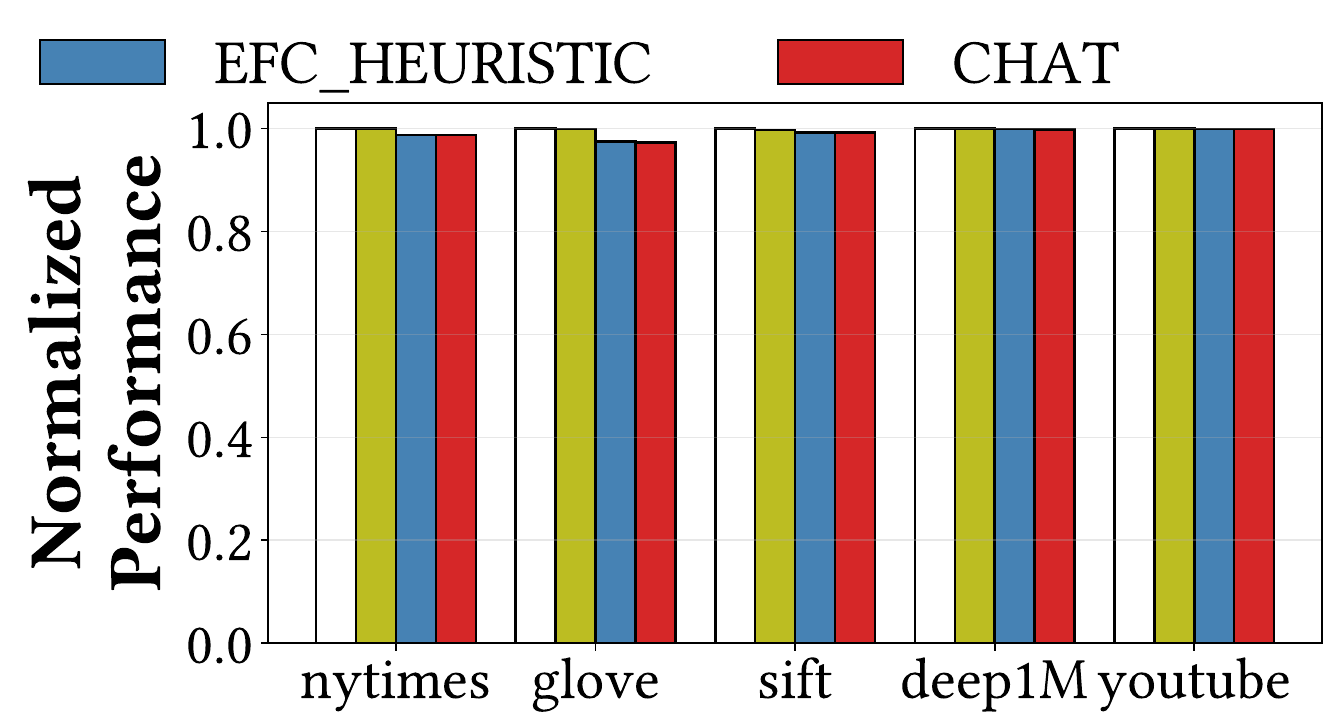}
    \vspace{-1.5em}
    \caption{}
    \label{fig:ablation_sub_b}
  \end{subfigure}
  \vspace{-1em}
  \caption{Impact of heuristic search-space reduction. (a) Tuning time. (b) Average performance (Recall and QPS). All values are normalized to BASE.}
  \label{fig:ablation_figure}
\end{figure}

\vspace{-0.2em}
\subsection{Impact of Heuristic Search Space Reduction}

To assess the effectiveness of the two heuristics used in \frameworkname{} for narrowing the search ranges of \efs{} and \efc{}, we performed an ablation study comparing four configurations: (1) \frameworkname{}, which incorporates both heuristics described in Sections~\ref{sec:M_efc_search} and \ref{sec:efs_search}; (2) EFS\_HEURISTIC, which applies only the \efs{} range narrowing heuristic; (3) EFC\_HEURISTIC, which applies only the \efc{} heuristic; and (4) BASE, which disables both heuristics and performs exhaustive search over the entire range using binary and ternary search procedures. The performance metric was defined as the aggregated objective value under a Recall threshold of 0.95 and a QPS threshold set at the 75th percentile.

\gun{
Figure~\ref{fig:ablation_figure}(a) shows that each heuristic individually reduces tuning time
relative to BASE, with the \efs{} heuristic achieving speedups of 1.13$\times$–4.96$\times$
and the \efc{} heuristic yielding 1.05$\times$–3.85$\times$ improvements.
When combined, \frameworkname{} attains speedups of 1.17$\times$–6.85$\times$,
indicating a complementary effect between the two heuristics.
Figure~\ref{fig:ablation_figure}(b) further shows that these gains do not compromise solution quality:
\frameworkname{} achieves normalized performance of 98.03\%–99.9\%, while EFS\_HEURISTIC
and EFC\_HEURISTIC attain 99.4\%–99.9\% and 98.7\%–99.9\%, respectively,
demonstrating that heuristic search space reduction substantially accelerates tuning
with only negligible loss in objective performance.}

\begin{figure}[t]
  \centering
  \includegraphics[width=0.95\linewidth]{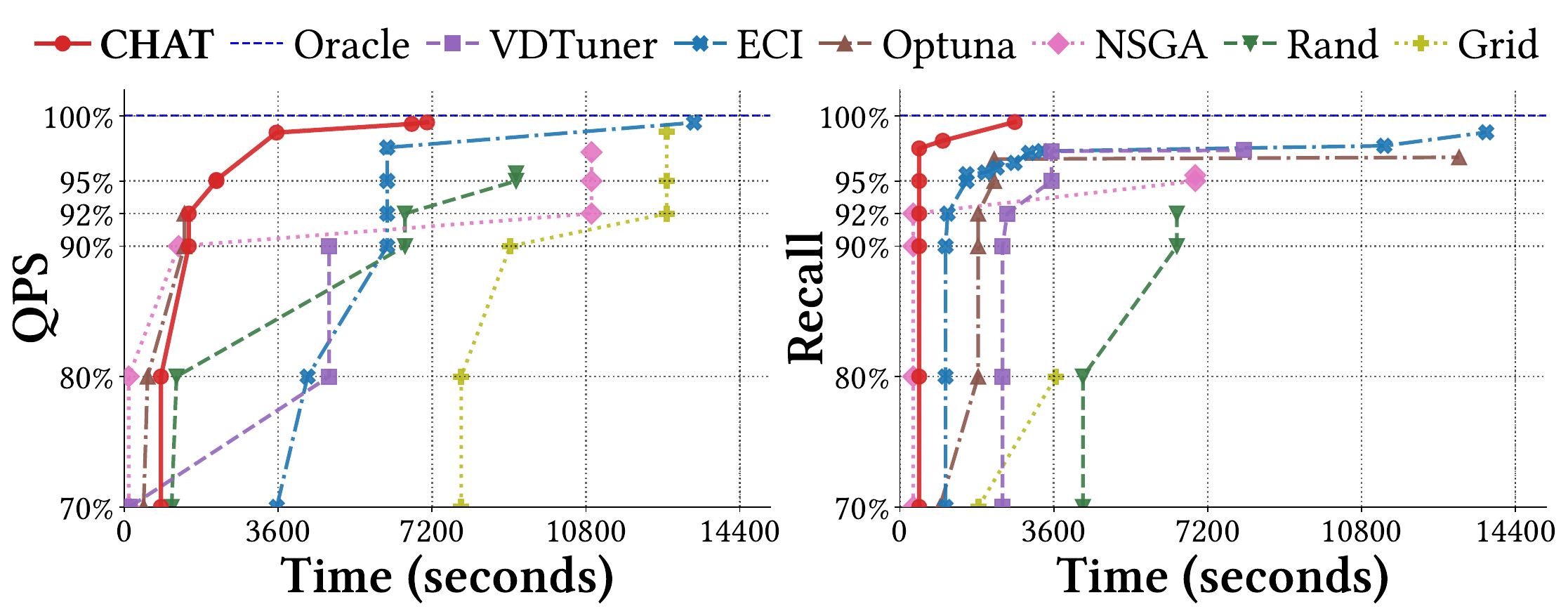}
  \vspace{-1em}
  \caption{Performance of \frameworkname{} and baselines on the \texttt{nytimes} dataset using Milvus.}
  \label{fig:milvus_figure}
\end{figure}

\subsection{Evaluation on a Commercial VectorDB}
\label{sec:milvus_section}

To validate the generality and practical effectiveness of our approach, we conducted additional experiments on Milvus, a widely adopted commercial vector database, beyond the open-source HNSW implementations evaluated earlier. We followed the same experimental setup described in Section~\ref{sec:setup}. Specifically, we applied a minimum Recall constraint of 0.95, or a minimum QPS constraint defined as the 75th percentile of all configurations for each dataset, with a fixed tuning time budget of four hours.

As shown in Figure~\ref{fig:milvus_figure}, which presents results on the \texttt{nytimes} dataset, 
our method consistently outperformed all baseline methods under both types of constraints%
\footnote{We observed similar trends on the other datasets used in this experiment.}. 
Under the Recall constraint, it achieved up to 99.2\% and at least 96.8\% of the Oracle's optimal QPS. 
Under the QPS constraint, it successfully matched the Oracle's optimal Recall. 
These results demonstrate that our method remains effective and stable even in production-grade
vector database environments, maintaining high performance and constraint satisfaction under
limited tuning budgets.

%% file: sections/related_work.tex
\section{Related Work}
\label{sec:related_work}

\textbf{VectorDB and ANN Indexing.}
Vector databases and approximate nearest neighbor (ANN) search methods are
central to high-dimensional data retrieval~\cite{10.14778/3476249.3476255,8681160}.
Early methods such as locality-sensitive hashing (LSH)~\cite{indyk1998approximate,10.14778/2850469.2850470}
and product quantization (PQ)~\cite{5432202,pqarticle} introduced efficient
hashing and compression techniques. Graph-based ANN indices, including
Hierarchical Navigable Small World (HNSW)~\cite{malkov2018efficientrobustapproximatenearest}
and Navigating Spreading-out Graph (NSG)~\cite{fu2025fastapproximatenearestneighbor},
improve recall--speed trade-offs by organizing points into navigable proximity
graphs~\cite{9383170}. DiskANN~\cite{10.5555/3454287.3455520} scales
graph-based search to billion-scale datasets through memory--disk optimized
storage and retrieval. Libraries such as Faiss~\cite{8733051} and
ScaNN~\cite{guo2020acceleratinglargescaleinferenceanisotropic} combine
quantization and GPU acceleration for large-scale search, with adoption in
systems such as Milvus~\cite{10.1145/3448016.3457550}.

\textbf{Database Auto-Tuning.}
Automated database tuning has progressed from heuristic methods to
machine-learning-based approaches. OtterTune~\cite{vanaken2017automatic}
introduced supervised learning for knob tuning, while QTune~\cite{li2019qtune}
extended this direction with reinforcement learning (RL) for query-level
precision. Recent works emphasize constraint-awareness and multi-objective
optimization.

%% file: sections/conclusion.tex
\section{Conclusion}
\label{sec:conclusion}


\cgcg{
We presented \frameworkname{}, an API-level black-box but 
HNSW-structure-aware tuning framework that exploits monotone 
search-time feasibility boundaries, dominant unimodal trends over 
construction parameters, and separable resource dependencies to guide 
a hierarchical search with resource surrogates. Across datasets and 
HNSW backends, \frameworkname{} efficiently identifies and validates 
configurations that satisfy user-d efined performance and resource 
constraints, offering a practical path toward predictable tuning for 
production HNSW deployments.
}


%% file: appendix.tex
\setlength{\parskip}{0pt}

\newcommand*\filledcircled[1]{\tikz[baseline=(char.base)]{
  \node[shape=circle,fill=black,text=white,inner sep=1pt] (char) {\sffamily\bfseries\small #1};}}
\AtBeginDocument{%
  \providecommand\BibTeX{{%
    Bib\TeX}}}

\definecolor{comportableColor}{HTML}{2370CD}

\definecolor{goodGreen}{HTML}{006E51}
\newcommand{\ra}[1]{\textcolor{black}{#1}}
\newcommand{\rb}[1]{\textcolor{black}{#1}}
\newcommand{\rc}[1]{\textcolor{black}{#1}}




\acmConference[SIGMOD '27]{Proceedings of the 2027 International Conference on Management of Data}{June 13--19, 2027}{Huntington Beach, CA, USA}
\acmBooktitle{Proceedings of the 2027 International Conference on Management of Data (SIGMOD '27), June 13--19, 2027, Huntington Beach, CA, USA}

\acmISBN{978-1-4503-XXXX-X/2026/05}

\settopmatter{printacmref=false} 
\renewcommand\footnotetextcopyrightpermission[1]{} 

\section{\ra{Index Construction Algorithms of HNSW}}
\label{appendix:hnsw_index_build_algorithms}

\algrenewcommand\algorithmiccomment[1]{// #1}

\subsection{Data Insertion}
\label{appendix:hnsw_data_insertion_algo}

\begin{algorithm}[h!]
\caption{INSERT($hnsw, q, M, M_{\max},$ \textit{efc}$, m_L$)}
\label{alg:hnsw_insert_algo}
\footnotesize
\begin{tabbing}
00\quad\=00\quad\=00\quad\=00\quad\= \kill
\textbf{Input:} multilayer graph $hnsw$, new element $q$, \\ 
\>number of established connections $M$, \\
\>maximum number of connections for each element per layer $M_{\max}$, \\ 
\>size of the dynamic candidate list \textit{efc}, \\
\>normalization factor for level generation $m_L$ \\
\textbf{Output:} update $hnsw$ inserting element $q$ \\
\\
1 \>$W \leftarrow \emptyset$ \\
2 \>$ep \leftarrow$ get enter point for $hnsw$ \\
3 \>$L \leftarrow$ level of $ep$ \\
4 \>$l \leftarrow \lfloor -\ln(\mathrm{unif}(0..1)) \cdot m_L \rfloor$ \\
5 \>\textbf{for} $l_c \leftarrow L \ \ldots\ l+1$ \\
6 \>\>$W \leftarrow \mathrm{SEARCH\mbox{-}LAYER}(q, ep,$ \textit{efc}${=}1, l_c)$ \\
7 \>\>$ep \leftarrow$ get the nearest element from $W$ to $q$ \\
8 \>\textbf{for} $l_c \leftarrow \min(L,l) \ \ldots\ 0$ \\
9 \>\>$W \leftarrow \mathrm{SEARCH\mbox{-}LAYER}(q, ep,$ \textit{efc}$, l_c)$ \\
10 \>\>$neighbors \leftarrow \mathrm{SELECT\mbox{-}NEIGHBORS}(q, W, M, l_c)$ \\
11 \>\>add bidirectional connections from $neighbors$ to $q$ at layer $l_c$ \\
12 \>\>\textbf{for each} $e \in neighbors$ \\
13 \>\>\>$eConn \leftarrow neighbourhood(e)$ at layer $l_c$ \\
14 \>\>\>\textbf{if} $|eConn| > M_{\max}$ \\
15 \>\>\>\>$eNewConn \leftarrow \mathrm{SELECT\mbox{-}NEIGHBORS}(e, eConn, M_{\max}, l_c)$ \\
16 \>\>\>\>set $neighbourhood(e)$ at layer $l_c$ to $eNewConn$ \\
17 \>\>$ep \leftarrow W$ \\
18 \>\textbf{if} $l > L$ \\
19 \>\>set enter point for $hnsw$ to $q$
\end{tabbing}
\end{algorithm}

Algorithm \ref{alg:hnsw_insert_algo} describes the incremental insertion procedure used to construct the HNSW index. Each data point is inserted sequentially into a multi-layer proximity graph, where the number of layers and the connectivity at each layer are determined probabilistically and by fixed degree constraints.

For each inserted element $q$, the algorithm first samples a maximum layer $\ell$ from an exponentially decaying distribution controlled by the normalization parameter $m_L$. This mechanism ensures that the expected number of layers in the structure is $O(\log N)$, yielding a hierarchy in which higher layers contain progressively fewer nodes and serve as long-range routing shortcuts.

The insertion proceeds in two phases. In the first phase, the algorithm performs a greedy descent from the current entry point, starting at the topmost layer and continuing down to layer $\ell + 1$. At each of these layers, SEARCH-LAYER is executed with $ef = 1$, producing a single closest candidate that becomes the entry point for the next lower layer. This phase localizes the insertion position coarsely and incurs minimal search cost.

In the second phase, for layers $\min(L, \ell)$ down to the base layer, SEARCH-LAYER is executed with a larger candidate list size \textit{efc}. This broader search identifies a set of candidate neighbors that are then filtered by SELECT-NEIGHBORS to retain at most $M$ connections. The resulting neighbors are connected bidirectionally to $q$, and degree constraints are enforced by reapplying SELECT-NEIGHBORS to existing nodes whose neighborhoods exceed the layer-specific maximum degree.
This insertion procedure ensures that each node maintains bounded degree while preserving the navigability of the graph. Importantly, the construction cost is dominated by the repeated execution of SEARCH-LAYER across layers, where the breadth of exploration is controlled by \textit{efc} and the per-node scan cost is bounded by the realized degree. These properties directly motivate the resource-usage analysis in Section 4.

\subsection{Search Layers}
\label{appendix:hnsw_search_layers_algo}

\begin{algorithm}[h!]
\caption{SEARCH-LAYER($q, ep,$ \textit{efs}$, l_c$)}
\label{alg:hnsw_search_layer_algo}
\footnotesize
\begin{tabbing}
00\quad\=00\quad\=00\quad\=00\quad\= \kill
\textbf{Input:} query element $q$, enter points $ep$, \\
\>number of nearest to $q$ elements to return \textit{efs}, layer number $l_c$ \\
\textbf{Output:} \textit{efs} closest neighbors to $q$ \\
\\
1 \>$v \leftarrow epc$ \\
2 \>$C \leftarrow epc$ \\
3 \>$W \leftarrow epc$ \\
4 \>\textbf{while} $|C| > 0$ \\
5 \>\>$c \leftarrow$ extract nearest element from $C$ to $q$ \\
6 \>\>$f \leftarrow$ get furthest element from $W$ to $q$ \\
7 \>\>\textbf{if} $distance(c,q) > distance(f,q)$ \\
8 \>\>\>\textbf{break} \\
9 \>\>\textbf{for each} $e \in neighbourhood(c)$ at layer $l_c$ \\
10 \>\>\>\textbf{if} $e \notin v$ \\
11 \>\>\>\>$v \leftarrow v \cup e$ \\
12 \>\>\>\>$f \leftarrow$ get furthest element from $W$ to $q$ \\
13 \>\>\>\textbf{if} $distance(e,q) < distance(f,q)$ \textbf{or} $|W| <$ \textit{efs} \\
14 \>\>\>\>$C \leftarrow C \cup e$ \\
15 \>\>\>\>$W \leftarrow W \cup e$ \\
16 \>\>\>\textbf{if} $|W| >$ \textit{efs} \\
17 \>\>\>\>remove furthest element from $W$ to $q$ \\
18 \>\textbf{return} $W$
\end{tabbing}
\end{algorithm}

Algorithm \ref{alg:hnsw_search_layer_algo} defines the SEARCH-LAYER procedure, which serves as the fundamental primitive for both index construction and query processing in HNSW. Given a query element $q$, an entry point $ep$, a layer index $\ell$, and a candidate list size \textit{efs}, the procedure performs a best-first graph traversal restricted to layer $\ell$.

The algorithm maintains two priority queues: a candidate queue $C$, which stores nodes to be expanded, and a working set $W$, which stores the current best \textit{efs} candidates found so far. At each iteration, the closest element $c$ in $C$ is extracted and compared against the worst element in $W$. If no improvement is possible, the search terminates early.

Otherwise, the algorithm scans the neighbors of $c$ at layer $\ell$. For each previously unvisited neighbor, a distance evaluation is performed, and the neighbor is inserted into $C$ and $W$ if it improves the current candidate set. When $W$ exceeds size \textit{efs}, the farthest element is removed, maintaining a bounded working set.

This greedy best-first traversal ensures that SEARCH-LAYER explores only a limited subset of the graph, with the total number of expansions controlled by \textit{efs}. As a result, increasing \textit{efs} monotonically enlarges the exploration budget, while the cost of each expansion is bounded by the degree constraints imposed during construction. This monotonic dependence underlies the performance and resource trade-offs analyzed in Sections 3 and 4.

\subsection{Search Neighborhoods to Connect}
\label{appendix:hnsw_search_neighbor_algo}

\begin{algorithm}[h!]
\caption{SELECT-NEIGHBORS($q, C, M, l_c, extendCandidates,$\\$keepPrunedConnections$)}
\label{alg:hnsw_select_neighbors_algo}
\footnotesize
\begin{tabbing}
00\quad\=00\quad\=00\quad\=00\quad\=00\quad\= \kill
\textbf{Input:} base element $q$, candidate elements $C$, \\
\>number of neighbors to return $M$, layer number $l_c$, \\
\>flag indicating whether or not to extend candidate list $extendCandidates$, \\
\>flag indicating whether or not to add discarded elements $keepPrunedConnections$ \\
\textbf{Output:} $M$ elements selected by the heuristic \\
\\
1 \>$R \leftarrow \emptyset$ \\
2 \>$W \leftarrow C$ \\
3 \>\textbf{if} $extendCandidates$ \\
4 \>\>\textbf{for each} $e \in C$ \\
5 \>\>\>\textbf{for each} $e_{adj} \in neighbourhood(e)$ at layer $l_c$ \\
6 \>\>\>\>\textbf{if} $e_{adj} \notin W$ \\
7 \>\>\>\>\>$W \leftarrow W \cup e_{adj}$ \\
8 \>$W_d \leftarrow \emptyset$ \\
9 \>\textbf{while} $|W| > 0$ \textbf{and} $|R| < M$ \\
10 \>\>$e \leftarrow$ extract nearest element from $W$ to $q$ \\
11 \>\>\textbf{if} $e$ is closer to $q$ compared to any element from $R$ \\
12 \>\>\>$R \leftarrow R \cup e$ \\
13 \>\>\textbf{else} \\
14 \>\>\>$W_d \leftarrow W_d \cup e$ \\
15 \>\textbf{if} $keepPrunedConnections$ \\
16 \>\>\textbf{while} $|W_d| > 0$ \textbf{and} $|R| < M$ \\
17 \>\>\>$R \leftarrow R \cup$ extract nearest element from $W_d$ to $q$ \\
18 \>\textbf{return} $R$
\end{tabbing}
\end{algorithm}

Algorithm \ref{alg:hnsw_select_neighbors_algo} specifies the SELECT-NEIGHBORS using heuristic to search, which enforces bounded degree while maintaining graph navigability. Given a candidate set $C$, the procedure iteratively selects neighbors based not only on proximity to the target node $q$, but also on diversity among selected neighbors.

At each step, the closest remaining candidate is selected if it is sufficiently closer to $q$ than any previously selected neighbor. Candidates that fail this condition are temporarily discarded but may be reconsidered if the optional "keep pruned connections" flag is enabled. The process continues until $M$ neighbors are selected or the candidate pool is exhausted.

This heuristic induces two important structural effects. First, it limits the realized out-degree even when the degree cap increases, leading to saturation in average out-degree as $M$ grows. Second, it favors diverse connections that improve long-range navigability, which directly impacts search efficiency. These effects explain why the realized connectivity grows sublinearly with $M$ and why build time exhibits diminishing returns, as discussed in Section 4.

\section{\rc{Proof Sketches for Unimodality Lemmas}}
\label{app:proofs}

The lemmas in Sections~3.3 and~3.4 give sufficient conditions for
discrete unimodality on the ordered grids used by \frameworkname{}.
They are not universal guarantees. The purpose of the proof sketches
below is to show that, under the stated marginal benefit--cost
conditions, the corresponding objective sequence has no
negative-to-positive marginal reversal. Section~6.4 gives empirical
stress cases where these sufficient conditions fail.

\subsection{Proof Sketch for Lemma~3.1}
\label{app:proof-efc}

Fix \(M\) and an active constraint
\(C \in \{\mathrm{Recall}, \mathrm{QPS}\}\). Consider the ordered
construction-effort grid
\[
  \mathit{efc}_1 < \mathit{efc}_2 < \cdots < \mathit{efc}_n .
\]
For each \(\mathit{efc}_i\), let \(f_M(\mathit{efc}_i)\) denote the
constrained objective after choosing the boundary \(\mathit{efs}\)
that satisfies \(C\). Define the discrete marginal change
\[
  \Delta f_M(\mathit{efc}_i)
  =
  f_M(\mathit{efc}_{i+1}) - f_M(\mathit{efc}_i)
  =
  \Delta B_M(\mathit{efc}_i) - \Delta T_M(\mathit{efc}_i),
\]
where \(\Delta B_M\) is the marginal graph-quality benefit and
\(\Delta T_M\) is the residual traversal/search cost.

By construction, reductions in expanded nodes caused by better graph
navigability are counted as graph-quality benefit, not as negative
cost. Hence the residual cost term satisfies
\[
  \Delta T_M(\mathit{efc}_i) \ge 0 .
\]
If \(\Delta B_M(\mathit{efc}_i) > \Delta T_M(\mathit{efc}_i)\), then
\(\Delta f_M(\mathit{efc}_i) > 0\), so increasing \(\mathit{efc}\)
improves the constrained objective at this grid step. If
\(\Delta B_M(\mathit{efc}_i) \le \Delta T_M(\mathit{efc}_i)\), then
\(\Delta f_M(\mathit{efc}_i) \le 0\), so the objective no longer
improves at this step.

Condition~(ii) of Lemma~3.1 rules out a later recovery after this
crossing: once the marginal benefit no longer exceeds the residual
cost, larger \(\mathit{efc}\) values do not reveal a delayed
high-utility connectivity regime that restores
\(\Delta B_M > \Delta T_M\). Therefore, after the first non-positive
marginal step, all later marginal steps remain non-positive. The
sequence \(f_M(\mathit{efc}_1),\ldots,f_M(\mathit{efc}_n)\) therefore
has no negative-to-positive marginal reversal. It is non-decreasing
up to some peak grid point and non-increasing thereafter, allowing
ties. Thus, \(f_M(\mathit{efc})\) is unimodal on the ordered grid.

If condition~(ii) fails, a delayed graph-quality gain can make the
marginal benefit exceed the residual cost again after an earlier
decline, producing a second peak. This is precisely the failure mode
examined in Section~6.4.

\subsection{Proof Sketch for Lemma~3.2}
\label{app:proof-m}

The proof for the outer objective follows the same marginal-crossing
argument, but it applies to the best constrained objective obtainable
after re-optimizing \(\mathit{efc}\) for each \(M\). Consider the
ordered degree grid
\[
  M_1 < M_2 < \cdots < M_n .
\]
For each \(M_i\), define the outer envelope
\[
  g(M_i) = \max_{\mathit{efc}} F_C(M_i,\mathit{efc}),
\]
where \(F_C(M_i,\mathit{efc})\) is the constrained objective after
choosing the boundary \(\mathit{efs}\) satisfying the active
constraint \(C\). Let
\[
  \Delta g(M_i)
  =
  g(M_{i+1}) - g(M_i)
  =
  \Delta B_{\mathrm{out}}(M_i)
  -
  \Delta T_{\mathrm{out}}(M_i),
\]
where \(\Delta B_{\mathrm{out}}\) is the marginal graph-quality
benefit from increasing \(M\) with \(\mathit{efc}\) re-optimized, and
\(\Delta T_{\mathrm{out}}\) is the residual traversal/search cost.

As in Lemma~3.1, reductions in expanded nodes caused by better graph
navigability are counted as graph-quality benefit. Thus,
\[
  \Delta T_{\mathrm{out}}(M_i) \ge 0 .
\]
If
\(\Delta B_{\mathrm{out}}(M_i) \le \Delta T_{\mathrm{out}}(M_i)\),
then \(\Delta g(M_i) \le 0\). Condition~(ii) of Lemma~3.2 prevents a
within-branch recovery: after additional degree no longer provides
enough marginal graph-quality benefit to offset residual traversal
cost, larger \(M\) values do not by themselves restore a positive
marginal gain.

The outer-envelope case also requires an additional condition because
\(g(M)\) re-optimizes \(\mathit{efc}\) at every \(M\). Even if each
fixed-\(M\) slice is unimodal, re-optimizing \(\mathit{efc}\) could
create a new high-performing branch at larger \(M\). Condition~(iii)
rules out this cross-branch recovery by requiring \(M\) and
\(\mathit{efc}\) to remain partial substitutes for graph quality:
larger \(M\) should not make a substantially different
\(\mathit{efc}\) branch suddenly dominate after the outer objective
has already entered the overhead-dominated regime.

Together, conditions~(ii) and~(iii) imply that once
\(\Delta g(M_i)\) becomes non-positive, all later marginal changes
remain non-positive. Hence the sequence
\(g(M_1),\ldots,g(M_n)\) has no negative-to-positive marginal
reversal. It is non-decreasing up to some peak grid point and
non-increasing thereafter, allowing ties. Therefore, \(g(M)\) is
unimodal over the ordered \(M\) grid.

If condition~(iii) fails, re-optimizing \(\mathit{efc}\) at larger \(M\)
may create a new high-performing branch and make the outer envelope
bimodal or multi-modal. Section~6.4 reports such stress cases and
shows that \frameworkname{} still returns the best-so-far validated
feasible configuration, although it is not a global optimizer for
arbitrary multi-modal objectives.

\section{Dataset}
\label{appendix:dataset}

\subsection{Dataset Statistics}
\label{appendix:dataset_statistics}

Table~\ref{tab:dataset_statistics} summarizes the datasets used in our evaluation.
Four datasets (\texttt{nytimes}, \texttt{glove}, \texttt{sift}, and \texttt{deep1M}) are drawn from
ANN Benchmarks.
The \texttt{youtube} dataset is constructed to emulate a realistic
video embedding retrieval workload.
The datasets span a wide range of dimensionalities (100--1024), corpus sizes (290K--1.2M),
and distance metrics, covering diverse operating regimes for HNSW hyperparameter tuning.

\begin{table}[h!]
\centering
\footnotesize
\begin{adjustbox}{max width=\linewidth}
\begin{tabular}{@{} lcccc @{}}
\toprule
\textbf{Dataset} & \textbf{Dimension} & \textbf{\# Data Points} & \textbf{\# Queries} & \textbf{Distance} \\
\midrule
\texttt{nytimes} & 256 & 290,000 & 10,000 & Angular \\
\texttt{glove}   & 100 & 1,183,514 & 10,000 & Angular \\
\texttt{sift}    & 128 & 1,000,000 & 10,000 & Euclidean \\
\texttt{deep1M}  & 256 & 1,000,000 & 1,000 & Angular \\
\texttt{youtube} & 1024 & 990,072 & 10,000 & Angular \\
\bottomrule
\end{tabular}
\end{adjustbox}
\caption{Dataset characteristics.}
\label{tab:dataset_statistics}
\end{table}

\subsection{\ra{Youtube dataset}}
\label{appendix:Youtube dataset}

\cgcg{
The \texttt{youtube} dataset is built from publicly available YouTube 
metadata of high-view-count videos in a late-2024 snapshot. Each vector 
corresponds to an extracted I-frame (keyframe), not to the whole video, 
so multiple vectors may share the same video\_id. We encode 
each I-frame into a 1024-dimensional vector using OpenCLIP 
H-14 (frames only; no audio, subtitles, or 
text metadata) and evaluate under angular distance. The final corpus 
contains 990{,}072 vectors with 10{,}000 held-out queries. We removed 
only items whose metadata contained unusable special characters or 
emojis; no other cleaning or vector post-processing was applied.
}

\section{Baselines}
\label{appendix:baselines}

\subsection{Details of Baselines}
\label{appendix:baseline_details}

We provide a brief overview of the baseline tuners used in our experiments.
All baselines are treated as black-box methods: they propose HNSW configurations $(\M{}, \efc{}, \efs{})$ and observe the resulting Recall, QPS, build time, and index size, without accessing internal graph statistics.

\paragraph{Grid Search and Random Search}
Grid Search evaluates configurations over a predefined discretized grid, while Random Search samples configurations uniformly from the same search space. 
These two methods serve as standard model-free baselines. 
Grid Search is deterministic but sensitive to grid granularity, whereas Random Search is simple and broadly exploratory but does not use past observations to guide future trials.

\paragraph{Optuna}
Optuna is a general-purpose model-based hyperparameter optimization 
framework that uses the Tree-structured Parzen Estimator (TPE) sampler 
to propose configurations from past observations. In our experiments, 
infeasible trials (those violating Recall or QPS targets) are penalized 
via a large-penalty scalar objective proportional to the degree of 
violation, so feasibility is enforced softly through the objective 
rather than explicit filtering. Optuna therefore represents a strong 
generic Bayesian optimization baseline, but it does not exploit 
HNSW-specific structural properties such as monotonicity or unimodality.

\paragraph{NSGA-II}
NSGA-II is a population-based evolutionary algorithm for multi-objective 
optimization. It maintains a diverse population of candidate solutions 
and evolves them to approximate a Pareto frontier over competing 
objectives. In our experiments, constraint violation is treated as one 
objective to minimize alongside maximizing the remaining objective 
(Recall or QPS), so feasibility is encouraged through Pareto dominance 
rather than explicit filtering. NSGA-II thus serves as a representative 
evolutionary baseline for HNSW tuning under the same Recall/QPS 
performance targets.

\paragraph{ECI}
ECI is a constrained Bayesian optimization baseline based on Expected Constrained Improvement. 
It models both objective performance and constraint feasibility, and selects configurations that are expected to improve the objective while satisfying the constraints. 
Compared with generic Bayesian optimization, ECI is more explicitly constraint-aware, making it a natural baseline for evaluating constraint-driven HNSW tuning.

\paragraph{VDTuner}
VDTuner is a Bayesian optimization-based tuning framework designed for vector database systems. 
It searches for high-quality configurations by modeling trade-offs between search quality and efficiency. 
We restrict VDTuner to the HNSW parameter space for a fair comparison, and evaluate it under the same Recall/QPS feasibility constraints and tuning-time budget as other methods.

\paragraph{Oracle Solution}
The Oracle Solution exhaustively evaluates the full predefined HNSW search space without a tuning-time budget and selects the best feasible configuration. 
It is not a practical tuner, but serves as an upper bound for measuring how close each method comes to the best achievable performance in our search space.

\subsection{\ra{Hyperparameters of Baselines}}
\label{appendix:baseline_hyperparameters}

Table~\ref{tab:baseline_hyperparameters} summarizes the key hyperparameter configurations
and experimental settings used for all baseline tuning methods evaluated in this study.

All stochastic methods were evaluated using the same random seed to ensure
fair and reproducible comparisons across baselines.

Unless otherwise stated, all baselines share the same search space over
$(M, efc, efs)$, enforce the hard constraint $M \le efc$, and are evaluated
under an identical wall-clock tuning budget. The table highlights the
method-specific differences in search strategies, discretization schemes,
and constraint-handling mechanisms.

\begin{table}[h!]
\centering
\footnotesize
\label{tab:baseline_hyperparameters}
\begin{adjustbox}{max width=\linewidth}
\begin{tabular}{lll}
\toprule
\textbf{Stage} & \textbf{Parameter} & \textbf{Value} \\
\midrule

Random Search & Search strategy & Uniform random sampling \\
 & Search space & $(M, efc, efs)$ \\
 & Parameter constraint & $M \le efc$ (enforced at sampling) \\
 & Discretization & None \\
 & Tuning budget & 14,400 seconds (4 hours) \\
\midrule

Grid Search & Search strategy & Exhaustive grid enumeration \\
 & $M$ step size & 4 \\
 & $efc$ step size & 64 \\
 & $efs$ step size & 128 \\
 & Parameter constraint & $M \le efc$ \\
 & Tuning budget & 14,400 seconds (or grid exhaustion) \\
\midrule

Optuna (TPE) & Search strategy & Bayesian optimization (TPE sampler) \\
 & Discretization & $\Delta M{=}1$, $\Delta efc{=}8$, $\Delta efs{=}16$ \\
 & Constraint handling & Large infeasible penalty \\
 & Objective & $-\text{QPS}$ or $-\text{Recall}$ \\
 & Tuning budget & 14,400 seconds \\
\midrule

NSGA-II & Search strategy & Evolutionary multi-objective search \\
 & Population size & 32 \\
 & Discretization & $\Delta M{=}1$, $\Delta efc{=}8$, $\Delta efs{=}16$ \\
 & Constraint handling & Violation-first objective \\
 & Tuning budget & 14,400 seconds \\
\midrule

ECI & Search strategy & Constrained Bayesian optimization \\
 & Warm-up samples & 16 \\
 & Discretization & $\Delta M{=}1$, $\Delta efc{=}8$, $\Delta efs{=}16$ \\
 & Feasibility model & Gaussian Process classifier \\
 & Acquisition & Constrained Expected Improvement \\
 & Tuning budget & 14,400 seconds \\
\midrule

VDTuner & Search strategy & Multi-objective Bayesian optimization \\
 & Acquisition & qEHVI \\
 & Initialization & Default configuration \\
 & Constraint handling & Implicit (CEI-style) \\
 & Tuning budget & 14,400 seconds \\
\bottomrule
\end{tabular}
\end{adjustbox}
\caption{Hyperparameter settings of baseline tuning methods used in our experiments.
All baselines were evaluated under the same wall-clock tuning budget of 14,400 seconds
(4 hours) per dataset.}
\label{tab:baseline_hyperparameters}
\end{table}

\section{\ra{Raw QPS Thresholds Used in the Evaluation}}    
\label{app:raw-qps-targets}

\begin{table}[!htbp]
\centering
\footnotesize
\setlength{\tabcolsep}{10pt}
\renewcommand{\arraystretch}{0.95}
\begin{tabular}{lrr}
\toprule
\textbf{Dataset} & \textbf{Faiss} & \textbf{Hnswlib} \\
\midrule
nytimes & 29,225 & 27,207 \\
glove & 49,940 & 37,410 \\
sift & 59,298 & 46,038 \\
deep1M & 24,039 & 27,150 \\
youtube & 12,909 & 9,795 \\
\bottomrule
\end{tabular}
\vspace{0.5em}
\caption{Raw QPS values corresponding to the 75th-percentile QPS
constraints used in Section~6.2. The percentile
is used only to make constraint difficulty comparable across workloads;
\frameworkname{} receives the raw QPS value shown here.}
\label{tab:raw-qps-targets}
\end{table}

For each dataset--backend pair, we computed the QPS target by measuring
QPS over the predefined HNSW search space and taking the 75th percentile
of the resulting distribution. This percentile is used only to make the
QPS constraint comparably challenging across workloads; after it is
computed, \frameworkname{} receives the threshold as a raw QPS constraint.
Table~\ref{tab:raw-qps-targets} reports these raw values.

\section{\rc{Analysis of Constraint Resource Models}}
\label{app:constraint_resource_models}

\begin{figure}[t]
    \centering
    \includegraphics[width=.95\linewidth]{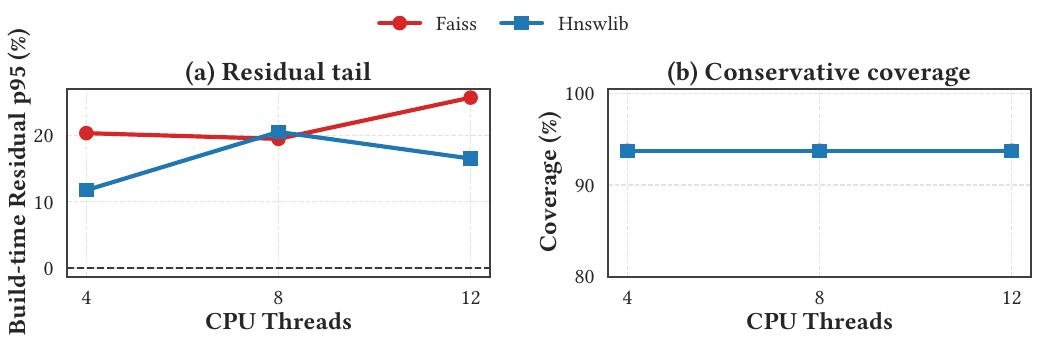}
    \caption{CPU-thread sensitivity of the resource feasibility test. We report the build-time underestimation residual tail and conservative-test coverage across CPU-thread counts for Faiss and Hnswlib.}
    \label{fig:appendix_constraint_model_backend_cpu_coverage}
\end{figure}

This appendix provides the resource-model diagnostics referenced in
Section~6.3. We evaluate the calibrated closed-form surrogates for build time
and index size from three perspectives: sensitivity to CPU parallelism,
calibration between predicted and measured resource usage, and residual
behavior across dataset scales and backends. Since the resource filter is used
to reject candidates before construction, the relevant error is one-sided
underestimation: cases where the model predicts a lower resource cost than the
measured value.

\begin{figure}[t]
    \centering
    \includegraphics[width=.95\linewidth]{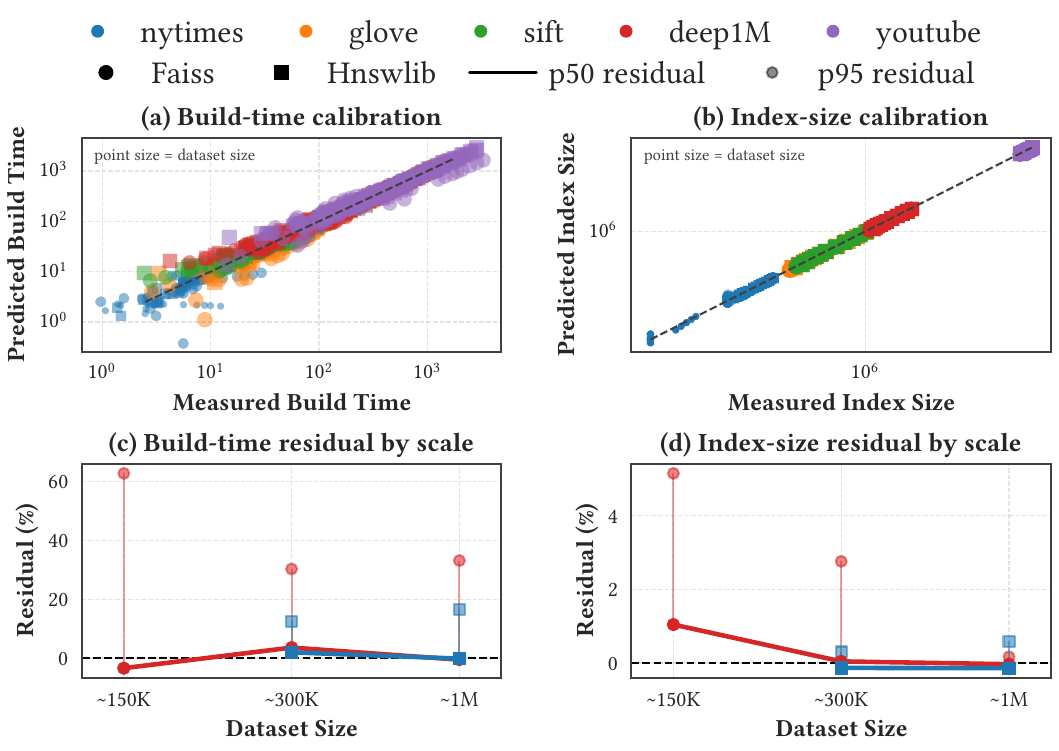}
\caption{\rc{Calibration and residual diagnostics for the closed-form resource models.
(a) and (b) compare predicted and measured build time and index size;
the dashed diagonal indicates perfect calibration, and marker size denotes
dataset scale. (c) and (d) report p50 and p95 one-sided
underestimation residuals for build time and index size across dataset scales,
datasets, and backends. Build-time residuals are larger and more variable,
whereas index-size residuals remain small because index size is dominated by
degree-cap-induced adjacency storage.}}
\label{fig:resource_model_diagnostics}
    \label{fig:appendix_constraint_model_calibration_residual}
\end{figure}

Figure~\ref{fig:appendix_constraint_model_backend_cpu_coverage} reports CPU-thread sensitivity.
Across the evaluated thread counts, the build-time underestimation residual
tail and conservative-test coverage remain stable for both Faiss and Hnswlib,
indicating that the online calibration and safety margin absorb backend and
parallelism effects. Figure~\ref{fig:resource_model_diagnostics} further shows
that the calibrated models track measured build time and index size across
dataset scales. Build-time residuals are larger and more variable because
construction time is affected by implementation and scheduling variability,
whereas index size is dominated by degree-cap-induced adjacency storage.
These diagnostics support using the compact resource models as conservative
feasibility filters rather than exact performance predictors.

\section{\rc{Synthetic Datasets for Structural Robustness Tests}}
\label{appendix:synthetic_datasets}

\cgcg{
This appendix describes the synthetic workloads used in
Section~6.4: six regular workloads that
vary common dataset characteristics, and one adversarial workload
designed to violate the diminishing-utility condition of
Lemmas~3.1 and~3.2.

\paragraph{Common methodology.}
Each dataset contains $N_{\text{train}} = 100{,}000$ training vectors
and $N_{\text{query}} = 1{,}000$ query vectors. For the six regular
workloads, training and query vectors are drawn from the same
distribution. For the adversarial workload, the query set is
intentionally biased toward local and bridge-requiring queries to
expose a delayed high-utility connectivity regime. All datasets use
$L_2$ distance with $k=10$, and ground truth is computed by exact
brute-force $k$-NN over the training set. All generators use a fixed
random seed of 42.

\paragraph{Well-separated.}
Eight isotropic Gaussian clusters in 128 dimensions, 12{,}500 points
per cluster, with $\sigma=0.5$. Cluster centers are sampled from
$[-10, 10]^{128}$ subject to a target nearest-center distance of at
least 8. This serves as a baseline case with cleanly separated
clusters.

\paragraph{Overlapping.}
Same as well-separated except $\sigma=2.0$ and centers are generated
with a target nearest-center distance of approximately 5, so that
cluster boundaries overlap substantially. This yields blurred local
neighborhoods and a harder $k$-NN problem.

\paragraph{Imbalanced.}
Eight clusters in 128 dimensions with skewed sizes
$\{50{,}000, 25{,}000, 10{,}000,$ $5{,}000, 4{,}000, 3{,}000, 2{,}000, 1{,}000\}$,
$\sigma=0.8$, target nearest-center distance 8. This tests robustness
to highly non-uniform cluster populations.

\paragraph{High-dimensional.}
Eight isotropic Gaussian clusters in 4096 dimensions with
$\sigma=1.0$, points equally distributed across clusters. This case
tests sensitivity to high dimensionality, where pairwise distances
concentrate.

\paragraph{Anisotropic.}
Eight Gaussian clusters in 128 dimensions with diagonal covariance
matrices whose entries are sampled from $[0.2, 5.0]$ independently
per cluster, producing elongated, axis-aligned clusters. This tests
robustness to non-spherical local geometry.

\paragraph{Uniform.}
A single uniform distribution over $[-10, 10]^{128}$ with no cluster
structure. We include it as a stress case where the index cannot
rely on dense cluster-local neighborhoods.
}



%% file: sample-base.bib
@String{Computing = "Computing" }

@String{Computer = "{IEEE} Computer" }

@article{pqarticle,
author = {Pan, Zhibin and Wang, Liangzhuang and Wang, Yang and Liu, Yuchen},
year = {2020},
month = {03},
pages = {},
title = {Product Quantization with Dual Codebooks for approximate Nearest Neighbor Search},
volume = {401},
journal = {Neurocomputing},
doi = {10.1016/j.neucom.2020.03.016}
}

@article{10.14778/2850469.2850470,
author = {Huang, Qiang and Feng, Jianlin and Zhang, Yikai and Fang, Qiong and Ng, Wilfred},
title = {Query-aware locality-sensitive hashing for approximate nearest neighbor search},
year = {2015},
issue_date = {September 2015},
publisher = {VLDB Endowment},
volume = {9},
number = {1},
issn = {2150-8097},
url = {https://doi.org/10.14778/2850469.2850470},
doi = {10.14778/2850469.2850470},
journal = {Proc. VLDB Endow.},
month = sep,
pages = {1–12},
numpages = {12}
}

@misc{gao2024retrievalaugmentedgenerationlargelanguage,
      title={Retrieval-Augmented Generation for Large Language Models: A Survey}, 
      author={Yunfan Gao and Yun Xiong and Xinyu Gao and Kangxiang Jia and Jinliu Pan and Yuxi Bi and Yi Dai and Jiawei Sun and Meng Wang and Haofen Wang},
      year={2024},
      eprint={2312.10997},
      archivePrefix={arXiv},
      primaryClass={cs.CL},
      url={https://arxiv.org/abs/2312.10997}, 
}

@misc{xu2025hallucinationinevitableinnatelimitation,
      title={Hallucination is Inevitable: An Innate Limitation of Large Language Models}, 
      author={Ziwei Xu and Sanjay Jain and Mohan Kankanhalli},
      year={2025},
      eprint={2401.11817},
      archivePrefix={arXiv},
      primaryClass={cs.CL},
      url={https://arxiv.org/abs/2401.11817}, 
}

@INPROCEEDINGS{1599245,
  author={Ngatchou, P. and Zarei, A. and El-Sharkawi, A.},
  booktitle={Proceedings of the 13th International Conference on, Intelligent Systems Application to Power Systems}, 
  title={Pareto Multi Objective Optimization}, 
  year={2005},
  volume={},
  number={},
  pages={84-91},
  doi={10.1109/ISAP.2005.1599245}}

@inproceedings{10.5555/2976040.2976144,
author = {Liu, Ting and Moore, Andrew W. and Gray, Alexander and Yang, Ke},
title = {An investigation of practical approximate nearest neighbor algorithms},
year = {2004},
publisher = {MIT Press},
address = {Cambridge, MA, USA},
booktitle = {Proceedings of the 18th International Conference on Neural Information Processing Systems},
pages = {825–832},
numpages = {8},
location = {Vancouver, British Columbia, Canada},
series = {NIPS'04}
}

@misc{guo2022manucloudnativevector,
      title={Manu: A Cloud Native Vector Database Management System}, 
      author={Rentong Guo and Xiaofan Luan and Long Xiang and Xiao Yan and Xiaomeng Yi and Jigao Luo and Qianya Cheng and Weizhi Xu and Jiarui Luo and Frank Liu and Zhenshan Cao and Yanliang Qiao and Ting Wang and Bo Tang and Charles Xie},
      year={2022},
      eprint={2206.13843},
      archivePrefix={arXiv},
      primaryClass={cs.DB},
      url={https://arxiv.org/abs/2206.13843}, 
}

@misc{bang2023multitaskmultilingualmultimodalevaluation,
      title={A Multitask, Multilingual, Multimodal Evaluation of ChatGPT on Reasoning, Hallucination, and Interactivity}, 
      author={Yejin Bang and Samuel Cahyawijaya and Nayeon Lee and Wenliang Dai and Dan Su and Bryan Wilie and Holy Lovenia and Ziwei Ji and Tiezheng Yu and Willy Chung and Quyet V. Do and Yan Xu and Pascale Fung},
      year={2023},
      eprint={2302.04023},
      archivePrefix={arXiv},
      primaryClass={cs.CL},
      url={https://arxiv.org/abs/2302.04023}, 
}

@misc{ma2025comprehensivesurveyvectordatabase,
      title={A Comprehensive Survey on Vector Database: Storage and Retrieval Technique, Challenge}, 
      author={Le Ma and Ran Zhang and Yikun Han and Shirui Yu and Zaitian Wang and Zhiyuan Ning and Jinghan Zhang and Ping Xu and Pengjiang Li and Wei Ju and Chong Chen and Dongjie Wang and Kunpeng Liu and Pengyang Wang and Pengfei Wang and Yanjie Fu and Chunjiang Liu and Yuanchun Zhou and Chang-Tien Lu},
      year={2025},
      eprint={2310.11703},
      archivePrefix={arXiv},
      primaryClass={cs.DB},
      url={https://arxiv.org/abs/2310.11703}, 
}

@misc{zhao2025surveylargelanguagemodels,
      title={A Survey of Large Language Models}, 
      author={Wayne Xin Zhao and Kun Zhou and Junyi Li and Tianyi Tang and Xiaolei Wang and Yupeng Hou and Yingqian Min and Beichen Zhang and Junjie Zhang and Zican Dong and Yifan Du and Chen Yang and Yushuo Chen and Zhipeng Chen and Jinhao Jiang and Ruiyang Ren and Yifan Li and Xinyu Tang and Zikang Liu and Peiyu Liu and Jian-Yun Nie and Ji-Rong Wen},
      year={2025},
      eprint={2303.18223},
      archivePrefix={arXiv},
      primaryClass={cs.CL},
      url={https://arxiv.org/abs/2303.18223}, 
}

@misc{openai2024gpt4technicalreport,
      title={GPT-4 Technical Report}, 
      author={OpenAI},
      year={2024},
      eprint={2303.08774},
      archivePrefix={arXiv},
      primaryClass={cs.CL},
      url={https://arxiv.org/abs/2303.08774}, 
}

@article{lewis2020retrieval,
  title={Retrieval-augmented generation for knowledge-intensive nlp tasks},
  author={Lewis, Patrick and Perez, Ethan and Piktus, Aleksandra and Petroni, Fabio and Karpukhin, Vladimir and Goyal, Naman and K{\"u}ttler, Heinrich and Lewis, Mike and Yih, Wen-tau and Rockt{\"a}schel, Tim and others},
  journal={Advances in Neural Information Processing Systems},
  volume={33},
  pages={9459--9474},
  year={2020}
}

@inproceedings{borgeaud2022improving,
  title={Improving language models by retrieving from trillions of tokens},
  author={Borgeaud, Sebastian and Mensch, Arthur and Hoffmann, Jordan and Cai, Trevor and Rutherford, Eliza and Millican, Katie and Van Den Driessche, George Bm and Lespiau, Jean-Baptiste and Damoc, Bogdan and Clark, Aidan and others},
  booktitle={International conference on machine learning},
  pages={2206--2240},
  year={2022},
  organization={PMLR}
}

@misc{izacard2022distillingknowledgereaderretriever,
      title={Distilling Knowledge from Reader to Retriever for Question Answering}, 
      author={Gautier Izacard and Edouard Grave},
      year={2022},
      eprint={2012.04584},
      archivePrefix={arXiv},
      primaryClass={cs.CL},
      url={https://arxiv.org/abs/2012.04584}, 
}

@misc{khandelwal2020generalizationmemorizationnearestneighbor,
      title={Generalization through Memorization: Nearest Neighbor Language Models}, 
      author={Urvashi Khandelwal and Omer Levy and Dan Jurafsky and Luke Zettlemoyer and Mike Lewis},
      year={2020},
      eprint={1911.00172},
      archivePrefix={arXiv},
      primaryClass={cs.CL},
      url={https://arxiv.org/abs/1911.00172}, 
}

@inproceedings{shi2023replugretrievalaugmentedblackboxlanguage,
  title={Replug: Retrieval-augmented black-box language models},
  author={Shi, Weijia and Min, Sewon and Yasunaga, Michihiro and Seo, Minjoon and James, Richard and Lewis, Mike and Zettlemoyer, Luke and Yih, Wen-tau},
  booktitle={Proceedings of the 2024 Conference of the North American Chapter of the Association for Computational Linguistics: Human Language Technologies (Volume 1: Long Papers)},
  pages={8371--8384},
  year={2024}
}

@article{malkov2018efficientrobustapproximatenearest,
  title={Efficient and robust approximate nearest neighbor search using hierarchical navigable small world graphs},
  author={Malkov, Yu A and Yashunin, Dmitry A},
  journal={IEEE transactions on pattern analysis and machine intelligence},
  volume={42},
  number={4},
  pages={824--836},
  year={2018},
  publisher={IEEE}
}

@article{AUMULLER2020101374,
title = {ANN-Benchmarks: A benchmarking tool for approximate nearest neighbor algorithms},
journal = {Information Systems},
volume = {87},
pages = {101374},
year = {2020},
issn = {0306-4379},
doi = {https://doi.org/10.1016/j.is.2019.02.006},
url = {https://www.sciencedirect.com/science/article/pii/S0306437918303685},
author = {Martin Aumüller and Erik Bernhardsson and Alexander Faithfull}
}

@misc{liashchynskyi2019gridsearchrandomsearch,
      title={Grid Search, Random Search, Genetic Algorithm: A Big Comparison for NAS}, 
      author={Petro Liashchynskyi and Pavlo Liashchynskyi},
      year={2019},
      eprint={1912.06059},
      archivePrefix={arXiv},
      primaryClass={cs.LG},
      url={https://arxiv.org/abs/1912.06059}, 
}

@misc{autotuningtheconstructionparameters,
author = {Zhou, Wenyang and Jiang, Yuzhi and Liu, Yingfan and Qiao, Xiaotian and Zhang, Hui and Li, Hui and Cui, Jiangtao},
year = {2024},
month = {01},
pages = {},
title = {Auto-Tuning the Construction Parameters of Hierarchical Navigable Small World Graphs},
doi = {10.2139/ssrn.4734062}
}

@inproceedings{yang2024vdtunerautomatedperformancetuning,
  title={Vdtuner: Automated performance tuning for vector data management systems},
  author={Yang, Tiannuo and Hu, Wen and Peng, Wangqi and Li, Yusen and Li, Jianguo and Wang, Gang and Liu, Xiaoguang},
  booktitle={2024 IEEE 40th International Conference on Data Engineering (ICDE)},
  pages={4357--4369},
  year={2024},
  organization={IEEE}
}

@inproceedings{10.1145/3448016.3457550,
author = {Wang, Jianguo and Yi, Xiaomeng and Guo, Rentong and Jin, Hai and Xu, Peng and Li, Shengjun and Wang, Xiangyu and Guo, Xiangzhou and Li, Chengming and Xu, Xiaohai and Yu, Kun and Yuan, Yuxing and Zou, Yinghao and Long, Jiquan and Cai, Yudong and Li, Zhenxiang and Zhang, Zhifeng and Mo, Yihua and Gu, Jun and Jiang, Ruiyi and Wei, Yi and Xie, Charles},
title = {Milvus: A Purpose-Built Vector Data Management System},
year = {2021},
isbn = {9781450383431},
publisher = {Association for Computing Machinery},
address = {New York, NY, USA},
url = {https://doi.org/10.1145/3448016.3457550},
doi = {10.1145/3448016.3457550},
booktitle = {Proceedings of the 2021 International Conference on Management of Data},
pages = {2614–2627},
numpages = {14},
location = {Virtual Event, China},
series = {SIGMOD '21}
}

@ARTICLE{8733051,
  author={Johnson, Jeff and Douze, Matthijs and Jégou, Hervé},
  journal={IEEE Transactions on Big Data}, 
  title={Billion-Scale Similarity Search with GPUs}, 
  year={2021},
  volume={7},
  number={3},
  pages={535-547},
  doi={10.1109/TBDATA.2019.2921572}}

@article{10.5555/2188385.2188395,
author = {Bergstra, James and Bengio, Yoshua},
title = {Random search for hyper-parameter optimization},
year = {2012},
issue_date = {3/1/2012},
publisher = {JMLR.org},
volume = {13},
number = {null},
issn = {1532-4435},
journal = {J. Mach. Learn. Res.},
month = feb,
pages = {281–305},
numpages = {25}
}

@article{Huang_2025,
   title={A Survey on Hallucination in Large Language Models: Principles, Taxonomy, Challenges, and Open Questions},
   volume={43},
   ISSN={1558-2868},
   url={http://dx.doi.org/10.1145/3703155},
   DOI={10.1145/3703155},
   number={2},
   journal={ACM Transactions on Information Systems},
   publisher={Association for Computing Machinery (ACM)},
   author={Huang, Lei and Yu, Weijiang and Ma, Weitao and Zhong, Weihong and Feng, Zhangyin and Wang, Haotian and Chen, Qianglong and Peng, Weihua and Feng, Xiaocheng and Qin, Bing and Liu, Ting},
   year={2025},
   month=jan, pages={1–55} }

@INPROCEEDINGS{6619223,
  author={Ge, Tiezheng and He, Kaiming and Ke, Qifa and Sun, Jian},
  booktitle={2013 IEEE Conference on Computer Vision and Pattern Recognition}, 
  title={Optimized Product Quantization for Approximate Nearest Neighbor Search}, 
  year={2013},
  volume={},
  number={},
  pages={2946-2953},
  doi={10.1109/CVPR.2013.379}}

@ARTICLE{5432202,
  author={Jégou, Herve and Douze, Matthijs and Schmid, Cordelia},
  journal={IEEE Transactions on Pattern Analysis and Machine Intelligence}, 
  title={Product Quantization for Nearest Neighbor Search}, 
  year={2011},
  volume={33},
  number={1},
  pages={117-128},
  doi={10.1109/TPAMI.2010.57}}

@ARTICLE{6809191,
  author={Muja, Marius and Lowe, David G.},
  journal={IEEE Transactions on Pattern Analysis and Machine Intelligence}, 
  title={Scalable Nearest Neighbor Algorithms for High Dimensional Data}, 
  year={2014},
  volume={36},
  number={11},
  pages={2227-2240},
  doi={10.1109/TPAMI.2014.2321376}}

@misc{youtubedataset,
  author = {dnotitia},
  title = {SeahorseDB-dataset},
  howpublished = {\url{https://huggingface.co/datasets/dnotitia/SeahorseDB-dataset/tree/main}},
  note = {Accessed: 2025-07-16}
}

@InProceedings{pmlr-v119-guo20h,
  title = 	 {Accelerating Large-Scale Inference with Anisotropic Vector Quantization},
  author =       {Guo, Ruiqi and Sun, Philip and Lindgren, Erik and Geng, Quan and Simcha, David and Chern, Felix and Kumar, Sanjiv},
  booktitle = 	 {Proceedings of the 37th International Conference on Machine Learning},
  pages = 	 {3887--3896},
  year = 	 {2020},
  editor = 	 {III, Hal Daumé and Singh, Aarti},
  volume = 	 {119},
  series = 	 {Proceedings of Machine Learning Research},
  month = 	 {13--18 Jul},
  publisher =    {PMLR},
  url = 	 {https://proceedings.mlr.press/v119/guo20h.html}
}

@misc{fu2025fastapproximatenearestneighbor,
      title={Fast Approximate Nearest Neighbor Search With The Navigating Spreading-out Graph}, 
      author={Cong Fu and Chao Xiang and Changxu Wang and Deng Cai},
      year={2025},
      eprint={1707.00143},
      archivePrefix={arXiv},
      primaryClass={cs.LG},
      url={https://arxiv.org/abs/1707.00143}, 
}

@inproceedings{10.1145/3292500.3330701,
author = {Akiba, Takuya and Sano, Shotaro and Yanase, Toshihiko and Ohta, Takeru and Koyama, Masanori},
title = {Optuna: A Next-generation Hyperparameter Optimization Framework},
year = {2019},
isbn = {9781450362016},
publisher = {Association for Computing Machinery},
address = {New York, NY, USA},
url = {https://doi.org/10.1145/3292500.3330701},
doi = {10.1145/3292500.3330701},
booktitle = {Proceedings of the 25th ACM SIGKDD International Conference on Knowledge Discovery \& Data Mining},
pages = {2623–2631},
numpages = {9},
location = {Anchorage, AK, USA},
series = {KDD '19}
}

@ARTICLE{996017,
  author={Deb, K. and Pratap, A. and Agarwal, S. and Meyarivan, T.},
  journal={IEEE Transactions on Evolutionary Computation}, 
  title={A fast and elitist multiobjective genetic algorithm: NSGA-II}, 
  year={2002},
  volume={6},
  number={2},
  pages={182-197},
  doi={10.1109/4235.996017}}

@inproceedings{indyk1998approximate,
author = {Indyk, Piotr and Motwani, Rajeev},
title = {Approximate nearest neighbors: towards removing the curse of dimensionality},
year = {1998},
isbn = {0897919629},
publisher = {Association for Computing Machinery},
address = {New York, NY, USA},
url = {https://doi.org/10.1145/276698.276876},
doi = {10.1145/276698.276876},
booktitle = {Proceedings of the Thirtieth Annual ACM Symposium on Theory of Computing},
pages = {604–613},
numpages = {10},
location = {Dallas, Texas, USA},
series = {STOC '98}
}

@inproceedings{vanaken2017automatic,
author = {Van Aken, Dana and Pavlo, Andrew and Gordon, Geoffrey J. and Zhang, Bohan},
title = {Automatic Database Management System Tuning Through Large-scale Machine Learning},
year = {2017},
isbn = {9781450341974},
publisher = {Association for Computing Machinery},
address = {New York, NY, USA},
url = {https://doi.org/10.1145/3035918.3064029},
doi = {10.1145/3035918.3064029},
booktitle = {Proceedings of the 2017 ACM International Conference on Management of Data},
pages = {1009–1024},
numpages = {16},
location = {Chicago, Illinois, USA},
series = {SIGMOD '17}
}

@article{li2019qtune,
author = {Lu, Jiaheng and Chen, Yuxing and Herodotou, Herodotos and Babu, Shivnath},
title = {Speedup your analytics: automatic parameter tuning for databases and big data systems},
year = {2019},
issue_date = {August 2019},
publisher = {VLDB Endowment},
volume = {12},
number = {12},
issn = {2150-8097},
url = {https://doi.org/10.14778/3352063.3352112},
doi = {10.14778/3352063.3352112},
journal = {Proc. VLDB Endow.},
month = aug,
pages = {1970–1973},
numpages = {4}
}

@misc{yu2020hyperparameteroptimizationreviewalgorithms,
      title={Hyper-Parameter Optimization: A Review of Algorithms and Applications}, 
      author={Tong Yu and Hong Zhu},
      year={2020},
      eprint={2003.05689},
      archivePrefix={arXiv},
      primaryClass={cs.LG},
      url={https://arxiv.org/abs/2003.05689}, 
}

@ARTICLE{7352306,
  author={Shahriari, Bobak and Swersky, Kevin and Wang, Ziyu and Adams, Ryan P. and de Freitas, Nando},
  journal={Proceedings of the IEEE}, 
  title={Taking the Human Out of the Loop: A Review of Bayesian Optimization}, 
  year={2016},
  volume={104},
  number={1},
  pages={148-175},
  doi={10.1109/JPROC.2015.2494218}}

@inproceedings{10.5555/2999325.2999464,
author = {Snoek, Jasper and Larochelle, Hugo and Adams, Ryan P.},
title = {Practical Bayesian optimization of machine learning algorithms},
year = {2012},
publisher = {Curran Associates Inc.},
address = {Red Hook, NY, USA},
booktitle = {Proceedings of the 26th International Conference on Neural Information Processing Systems - Volume 2},
pages = {2951–2959},
numpages = {9},
location = {Lake Tahoe, Nevada},
series = {NIPS'12}
}

@article{10.14778/3476249.3476255,
author = {Wang, Mengzhao and Xu, Xiaoliang and Yue, Qiang and Wang, Yuxiang},
title = {A comprehensive survey and experimental comparison of graph-based approximate nearest neighbor search},
year = {2021},
issue_date = {July 2021},
publisher = {VLDB Endowment},
volume = {14},
number = {11},
issn = {2150-8097},
url = {https://doi.org/10.14778/3476249.3476255},
doi = {10.14778/3476249.3476255},
journal = {Proc. VLDB Endow.},
month = jul,
pages = {1964–1978},
numpages = {15}
}

@ARTICLE{8681160,
  author={Li, Wen and Zhang, Ying and Sun, Yifang and Wang, Wei and Li, Mingjie and Zhang, Wenjie and Lin, Xuemin},
  journal={IEEE Transactions on Knowledge and Data Engineering}, 
  title={Approximate Nearest Neighbor Search on High Dimensional Data — Experiments, Analyses, and Improvement}, 
  year={2020},
  volume={32},
  number={8},
  pages={1475-1488},
  doi={10.1109/TKDE.2019.2909204}}

@ARTICLE{9383170,
  author={Fu, Cong and Wang, Changxu and Cai, Deng},
  journal={IEEE Transactions on Pattern Analysis and Machine Intelligence}, 
  title={High Dimensional Similarity Search With Satellite System Graph: Efficiency, Scalability, and Unindexed Query Compatibility}, 
  year={2022},
  volume={44},
  number={8},
  pages={4139-4150},
  doi={10.1109/TPAMI.2021.3067706}}

@inproceedings{gardner2014bayesian,
  title={Bayesian optimization with inequality constraints.},
  author={Gardner, Jacob R and Kusner, Matt J and Xu, Zhixiang Eddie and Weinberger, Kilian Q and Cunningham, John P},
  booktitle={ICML},
  volume={2014},
  pages={937--945},
  year={2014}
}

@article{10.5555/3013558.3013569,
author = {Wang, Ziyu and Hutter, Frank and Zoghi, Masrour and Matheson, David and De Freitas, Nando},
title = {Bayesian optimization in a billion dimensions via random embeddings},
year = {2016},
issue_date = {January 2016},
publisher = {AI Access Foundation},
address = {El Segundo, CA, USA},
volume = {55},
number = {1},
issn = {1076-9757},
journal = {J. Artif. Int. Res.},
month = jan,
pages = {361–387},
numpages = {27}
}

@misc{guo2020acceleratinglargescaleinferenceanisotropic,
      title={Accelerating Large-Scale Inference with Anisotropic Vector Quantization}, 
      author={Ruiqi Guo and Philip Sun and Erik Lindgren and Quan Geng and David Simcha and Felix Chern and Sanjiv Kumar},
      year={2020},
      eprint={1908.10396},
      archivePrefix={arXiv},
      primaryClass={cs.LG},
      url={https://arxiv.org/abs/1908.10396}, 
}

@article{wang2024steiner,
  title={Steiner-hardness: A query hardness measure for graph-based ann indexes},
  author={Wang, Zeyu and Wang, Qitong and Cheng, Xiaoxing and Wang, Peng and Palpanas, Themis and Wang, Wei},
  journal={Proceedings of the VLDB Endowment},
  volume={17},
  number={13},
  pages={4668--4682},
  year={2024},
  publisher={VLDB Endowment}
}

@article{fu12fast,
  title={Fast Approximate Nearest Neighbor Search With The Navigating Spreading-out Graph},
  author={Fu, Cong and Xiang, Chao and Wang, Changxu and Cai, Deng},
  journal={Proceedings of the VLDB Endowment},
  volume={12},
  number={5}
}

@inbook{10.5555/3454287.3455520,
author = {Subramanya, Suhas Jayaram and Devvrit and Kadekodi, Rohan and Krishaswamy, Ravishankar and Simhadri, Harsha Vardhan},
title = {DiskANN: fast accurate billion-point nearest neighbor search on a single node},
year = {2019},
publisher = {Curran Associates Inc.},
address = {Red Hook, NY, USA},
booktitle = {Proceedings of the 33rd International Conference on Neural Information Processing Systems},
articleno = {1233},
numpages = {11}
}

@software{ilharco_gabriel_2021_5143773,
  author       = {Ilharco, Gabriel and
                  Wortsman, Mitchell and
                  Wightman, Ross and
                  Gordon, Cade and
                  Carlini, Nicholas and
                  Taori, Rohan and
                  Dave, Achal and
                  Shankar, Vaishaal and
                  Namkoong, Hongseok and
                  Miller, John and
                  Hajishirzi, Hannaneh and
                  Farhadi, Ali and
                  Schmidt, Ludwig},
  title        = {OpenCLIP},
  month        = jul,
  year         = 2021,
  note         = {If you use this software, please cite it as below.},
  publisher    = {Zenodo},
  version      = {0.1},
  doi          = {10.5281/zenodo.5143773},
  url          = {https://doi.org/10.5281/zenodo.5143773}
}
